\documentclass[journal]{IEEEtran}
\IEEEoverridecommandlockouts
\usepackage{cite}
\usepackage{algorithmic}
\usepackage{algorithm}
\usepackage{amsmath,amssymb,amsfonts}
\usepackage{graphicx}
\usepackage{textcomp}
\usepackage{amsmath}
\usepackage{amsfonts}
\usepackage{caption}  
\usepackage{algorithm}

\usepackage{amsmath}

\def\BibTeX{{\rm B\kern-.05em{\sc i\kern-.025em b}\kern-.08em
    T\kern-.1667em\lower.7ex\hbox{E}\kern-.125emX}}

\usepackage[top=0.7in, bottom=1.03in, left=0.68in, right=0.68in]{geometry}

\usepackage[utf8]{inputenc}
\usepackage{multirow}
\usepackage{booktabs}
\usepackage{graphicx}
\usepackage{xcolor}  
\usepackage{makecell} 
\usepackage{booktabs}
\usepackage[table]{xcolor}
\usepackage{multirow}

\newcommand{\E}{\mathbb{E}}
\newcommand{\Var}{\mathrm{Var}}

\DeclareMathOperator*{\argmax}{arg\,max}

\newtheorem{rem}{Remark}
\newtheorem{lem}{Lemma}
    
\begin{document}

\title{Gamma-Distributed Geometric Constellation for ISAC: Design and Analysis}

\author{Amirhossein Keshavarzchafjiri,  Janith Kavindu Dassanayake,
Gayan~A.~Aruma~Baduge,~\IEEEmembership{Senior Member,~IEEE}, and Mojtaba Vaezi,~\IEEEmembership{Senior Member,~IEEE}

{\thanks{This article was presented in part at
the IEEE Global Communications Conference, December 2025 \cite{conference_version}.}}
\thanks{Amirhossein Keshavarzchafjiri and Mojtaba Vaezi are with the Department of Electrical and Computer Engineering, Villanova University, Villanova, PA 19085 USA (e-mail: \{akesha01, mvaezi\}@villanova.edu). Their work was supported by the U.S. National
Science Foundation  under Grant CCF-2326622.}
\thanks{Janith  Dassanayake and Gayan Baduge are with the School of Electrical, Computer, and Biomedical Engineering, Southern Illinois University, Carbondale, IL, USA, Email: \{janith.dassanayake, gayan.baduge\}@siu.edu. Their work in part has been supported  by the U.S. National Science Foundation under Grant CCF-2326621.}
}

\maketitle

\begin{abstract}

A novel Gamma-distributed geometric constellation design framework for integrated sensing and communication (ISAC) is proposed in this paper. In this framework, constellation points are modeled as samples drawn from a parameterized two-dimensional distribution, with a Gamma distribution for the amplitude and a uniform distribution for the phase. End-task performance metrics, namely, the probability of detection for sensing and mutual information for communication, are used as objective functions of the optimization problem, and the problem is solved via particle swarm optimization.
We further derive analytical performance bounds for the proposed design, including the union bound on the symbol error rate for communication and the Cramér–Rao bound for sensing parameter estimation. The proposed method is compared with constellations obtained via end-to-end neural network design, demonstrating competitive performance while requiring significantly fewer parameters and no training data. Moreover, the proposed geometric constellation is more compatible with conventional system architectures than probabilistic or neural network-based designs.
\end{abstract}


\section{Introduction}

Sixth-generation (6G) wireless communication systems is expected to integrate active sensing, i.e., radar,
 capabilities, enabling novel use cases in a wide range of applications, such as autonomous driving and smart cities~\cite{ liu2022isacwaveform, MSE_MIMO,vaezi2025ai}. 
In \textit{integrated sensing and communication} (ISAC), the conventional separation of sensing and communication tasks is replaced by a joint design approach that enables mutual benefits through synergy. This unified framework provides integration gain by sharing resources and coordination gain by enhancing the performance of each functionality through their interaction~\cite{holistic}.

ISAC requires unified waveform designs \cite{liu2018toward,dassanayake2025unsupervised} that jointly support sensing and communication. In this context, modulation constellation design is a key component for balancing these functionalities, as it directly influences both sensing and communication performance. A fundamental challenge arises from the inherent trade-off between the randomness that benefits communication and the deterministic characteristics preferred for optimal sensing~\cite{inner_bound}.

 However, traditional constellations, such as \emph{quadrature amplitude modulation} (QAM) and \emph{phase-shift keying} (PSK), usually excel at either communication efficiency or sensing accuracy, but not both simultaneously.  Therefore, new constellation shaping approaches are essential to balance sensing and communication requirements in ISAC systems. Constellation shaping can be broadly classified into \textit{geometric constellation shaping} (GCS) and \textit{probabilistic constellation shaping} (PCS). The former approach optimizes the location of constellation points  under the assumption of equal probability, whereas the latter retains conventional constellation structures but modifies the probability of transmitting each point~\cite{du2024pcs}. A more recent approach, known as joint shaping, simultaneously optimizes both the geometry and probabilities of the constellation points, offering greater flexibility~\cite{geiger2025joint}.

Recent research has explored various aspects of constellation shaping for ISAC. For instance,  focusing on finite-alphabet inputs, PCS-based constellation design is studied in~\cite{GS_constellation_design_journal}, where it is demonstrated that communication-oriented optimization leads to constellations that approximate a quasi-Gaussian distribution. In the context of orthogonal frequency-division multiplexing (OFDM)-based ISAC, a PCS-enabled signaling scheme is proposed in~\cite{du2024pcs}, where the statistical characteristics of the ambiguity function are controlled by optimizing the fourth moment of the constellation amplitudes. Similarly,  a PCS-based constellation design strategy is introduced in~\cite{yang2023random}, where the sidelobe level of the periodic autocorrelation function is reduced by adjusting the probability distribution of the modulation symbols.


 By investigating the practicality of PCS-based constellation design, valuable insights are provided in \cite{geiger2023experimental, experimental} through the implementation of a real-world ISAC testbed, with a particular focus on integrating PCS into an OFDM framework to jointly optimize communication efficiency and sensing performanc.  The sensing performance under mismatched filtering is further investigated in~\cite{temporal-frequency}, where the relationships among different sensing metrics, including output signal-to-noise ratio (SNR), integrated sidelobe level ratio, and the mean squared error of channel state information  estimation, are analyzed.

Besides conventional optimization methods, autoencoder-based frameworks have also been used to jointly optimize the geometric and probabilistic shaping of  constellations, to maximize generalized mutual information while satisfying sensing constraints (e.g., a target detection probability)~\cite{geiger2025joint}. For example, in~\cite{tang2024kl} the Kullback-Leibler (KL) divergence is employed as a key metric to jointly design the constellation and beamforming strategies in the presence of clutter. 
 Focusing on a bistatic ISAC setup, an end-to-end neural network (NN)-based constellation design is proposed in~\cite{learning_based}. This scheme does not require the receiver to know the transmitted symbols, and the impact of channel state information errors on both communication and sensing performance is investigated.

While probabilistic and data-driven constellation shaping have improved ISAC performance, they come with inherent drawbacks. PCS complicates transceiver design and overlooks geometric constraints, while NN-based approaches require large datasets  and incur high computational complexity~\cite{mateos2022end}.  To address these limitations, we propose a constellation design approach that models constellation points as samples drawn from a two-dimensional (2D) density function, with \textit{Gamma}-distributed amplitudes and \textit{uniformly} distributed phases, optimized to enhance both mutual information  for communication and probability of detection for sensing. In addition, to validate the performance of the proposed method, we analytically evaluate the designed constellations by deriving the \textit{union bound} \cite{proakis2008digital} on the symbol error rate (SER) for communication and the Cram\'{e}r--Rao bound (CRB) for sensing parameter estimation.

The main contributions of this work are as follows.
\begin{itemize}
\item We propose a novel framework for geometric constellation design  by parameterizing the constellation point distribution using only two parameters—the shape and scale of a Gamma distribution that governs the amplitudes of the constellation points. This formulation significantly reduces the optimization dimensionality. Unlike PCS-based approaches, our method is purely geometric and more amenable to hardware implementation, and in contrast to NN-based methods, it does not require extensive training data or high computational complexity, thereby offering improved interpretability.

    \item 
Rather than relying on surrogate objectives,
we directly guide the design using end-task metrics--empirical mutual information for communication and probability of detection for radar--thereby ensuring improvements in the desired performance measures. Also, we adopt the Albersheim approximation for radar sensing, which enables a closed-form evaluation of the detection probability.
\item We derive analytical expressions for the union bound on the SER in communication and the CRB for radar channel parameter estimation, and demonstrate that these bounds are reasonably tight. 
These results characterize the fundamental limits of the system and provide design insights for the proposed ISAC constellation design.

\end{itemize}

The remainder of this paper is organized as follows. Section~\ref{sec:system_model} describes the system model and the ISAC channel setup. Section~\ref{sec:design_methods} presents the proposed constellation design. Section~\ref{sec:evaluation} derives the analytical formulation of the SER union bound and the CRB for the estimation of the radar channel.  Section~\ref{sec:baseline} presents the baseline method based on NN. The numerical results are provided in Section~\ref{sec:results}, followed by the conclusions in Section~\ref{sec:conclusion}. Proofs and derivations are presented in Appendices.

\paragraph*{Notation}
Throughout this paper, lowercase letters (e.g., $x$) denote scalar variables, uppercase letters (e.g., $X$) represent either constants or random variables, and bold lowercase letters (e.g., $\mathbf{x}$) indicate vectors. Moreover, $f(\cdot)$, $p(\cdot)$, and $P$, respectively,  denote a probability density function (PDF), a probability mass function (PMF),  and a probability value. \(\mathbb{R}^+\) and \(\mathbb{C}\) denote the sets of positive real numbers and complex numbers, respectively. $\mathbb{E}[\cdot]$ and $\mathrm{Var}[\cdot]$ denote the expectation and variance operators, respectively.  $X\sim\mathcal{CN}(\mu, \sigma^2)$ denote that $X$ is a complex Gaussian  random variable with mean $\mu$ and variance $\sigma^2$.

\begin{figure*}[t]
    \centering
    \includegraphics[width=1.6\columnwidth, trim=50 0 50 0, clip]{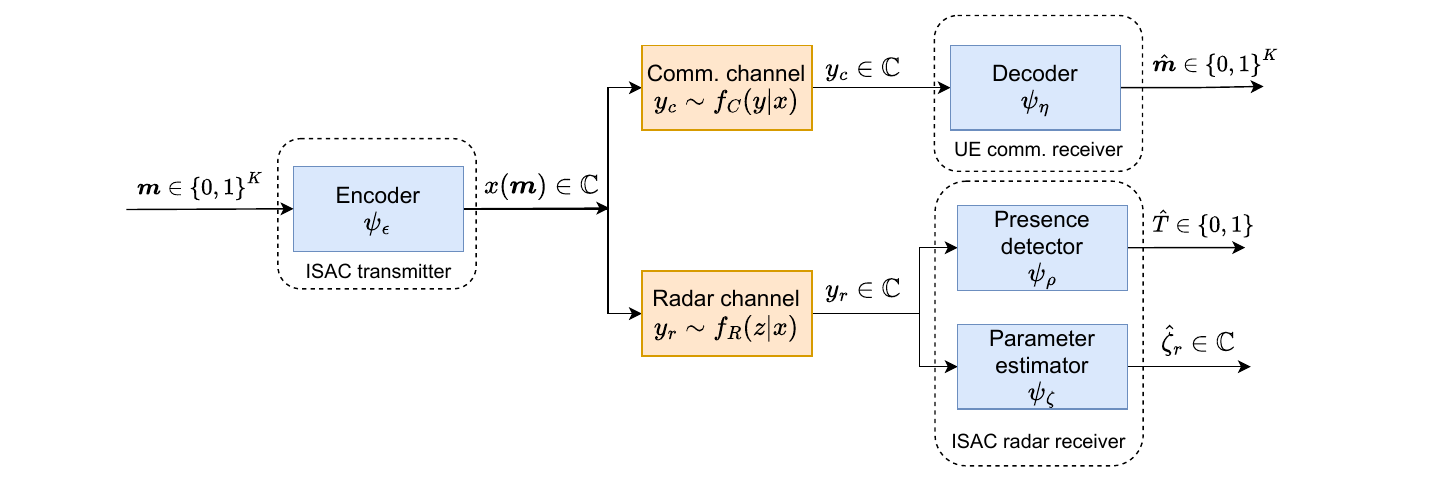}
    \caption{An ISAC system with
    one communication receiver and one target, where $\mathbf{m}$ is the message, $x(\boldsymbol{m})$ is the transmitted signal, and $y_c$ and $y_r$ are the received signals at the communication and sensing receivers, respectively.}
    \label{fig:block_diagram}
\end{figure*}


\section{System Model} \label{sec:system_model}
We consider an ISAC system in downlink with one communication receiver and one target.
Both the transmitter and the receiver are equipped with a single antenna. Figure~\ref{fig:block_diagram} shows the main components of the ISAC system in this setup.

Our goal is to develop a framework that optimally selects a subset \(\mathcal{S}\) of constellation points from a candidate set \(\bar{\mathcal{S}}\)  that satisfies the desired characteristics for both communication and sensing. The candidate set \(\bar{\mathcal{S}}\) is an \textit{amplitude phase shift keying} (APSK) structure, and its cardinality can be arbitrarily large. The rationale for selecting APSK as the candidate set is that, through appropriate selection, its subsets can realize a variety of desired constellations, including \textit{constant modulus} schemes such as PSK, constellations with large \textit{minimum distance} such as QAM, and multi-ring constellations.

We transmit a binary message vector of length $K$, $\boldsymbol{m} \in \{0,1\}^K$, using a modulation scheme defined by a subset $\mathcal{S} \subset \bar{\mathcal{S}}$ forming an $M$-ary constellation. The constellation set \( \mathcal{S} \) is defined as \( \{ x_i \in \mathbb{C} \mid i = 1, \dots, M \} \), where \( M = 2^K \). For any message \(\boldsymbol{m}\), \(x(\boldsymbol{m}) = x_i\), where \(i\) is the index corresponding to \(\boldsymbol{m}\) under the mapping scheme. A constellation symbol is also represented in polar coordinates as \( x = \rho e^{j\phi} \), where \( \rho \in \mathbb{R}^+ \) and \( \phi \in [0, 2\pi)\) denote the amplitude and phase, respectively. Similarly, for the set of constellation points, we have \(x_i = \rho_i e^{j\phi_i}\).

As shown in Fig.~\ref{fig:block_diagram}, the transmitted signal $x(\boldsymbol{m})$ propagates through the communication channel $f_C(y_c \mid x(\boldsymbol{m}))$ and is reflected back to the ISAC transmitter via the sensing channel $f_R(y_r \mid x(\boldsymbol{m}))$. The communication channel results in
\begin{align} \label{eq:communication channel}
y_c = x(\boldsymbol{m}) + n_c,
\end{align}
where \( n_c \sim \mathcal{CN}(0, \sigma^2_c) \) represents complex Gaussian noise, and the sensing channel results in 
\begin{align} \label{eq:radar channel}
y_r = \zeta_r x(\boldsymbol{m}) + n_r,
\end{align}
in which \(n_r \sim \mathcal{CN}(0,\sigma_r^2)\) denotes the channel noise, and \( \zeta_r\) denotes the complex reflection coefficient. 


\begin{figure*}[t]
    \centering

    \includegraphics[width=0.5\linewidth, trim=3cm 8cm 3cm 8cm, clip]{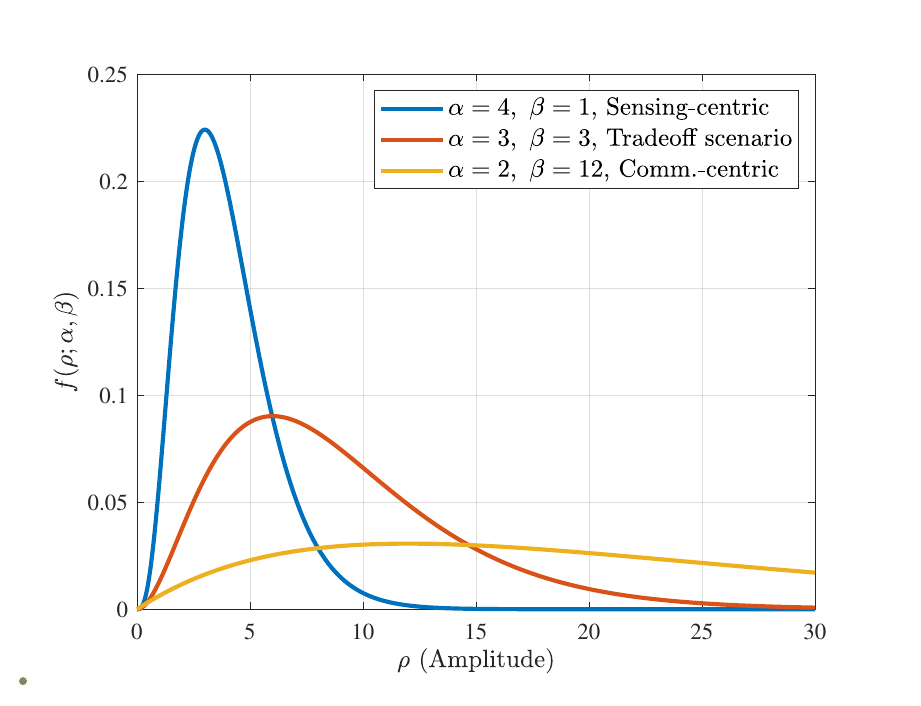}
    \hfill
    \includegraphics[width=0.48\linewidth, trim=3cm 8cm 3cm 8cm, clip]{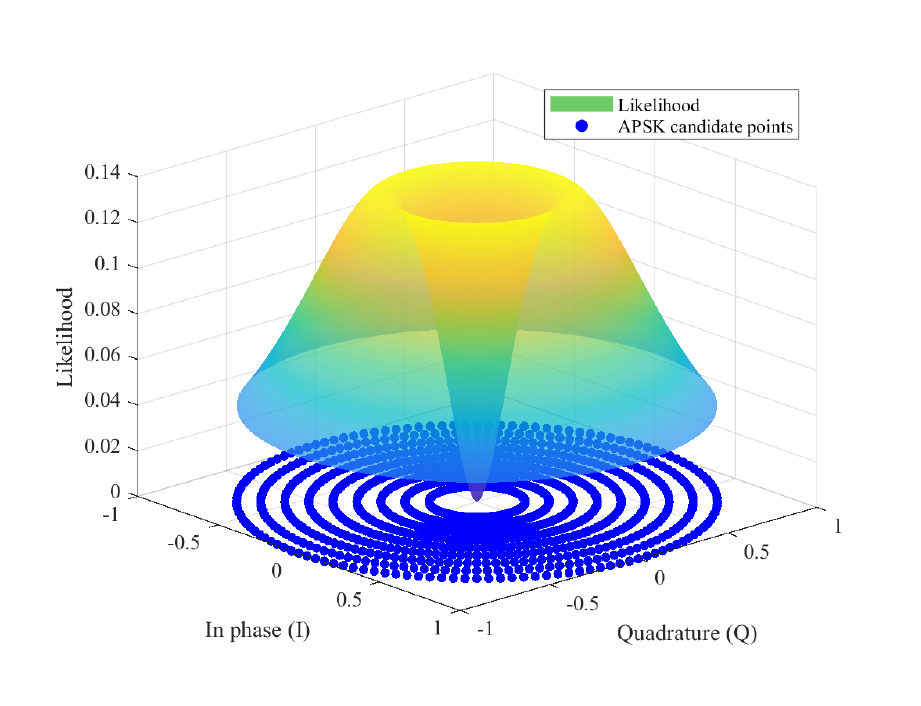}

    \par\vspace{0.2em}
    {\small(a) Gamma PDFs for different scenarios.}\hfill
    (b) APSK candidate points and an example likelihood function.

    \caption{Constellation points in the candidate set and their assigned likelihood that is Gamma-distributed in amplitude.}
    \label{fig:candidate_gamma}
\end{figure*}


\section{Proposed Constellation Design} \label{sec:design_methods}
\label{sec:Optimization problem}

From previous works on constellation design~\cite{yang2023random,du2024pcs}, it has been observed that the performance of an ISAC system is primarily influenced by how the constellation points are distributed in the constellation space. A well-performing distribution in communication-centric scenarios tends to produce a uniformly spread constellation, which effectively maximizes the minimum distance between points like QAM. In contrast, for radar-centric scenarios, such a constellation resembles a PSK-like structure. For intermediate trade-off scenarios, a multi-ring-like distribution appears to be optimal~\cite{du2024pcs}.




 In light of these observations, we model the constellation points as samples drawn from a 2D PDF in the constellation space. As a first step, an appropriate PDF must be chosen, i.e., one that is sufficiently expressive to model the desired variety of constellation structures. Since communication and sensing performance are invariant to constellation rotation \cite{proakis2008digital}, we can assume the phase component $\phi$ of constellation points to be uniformly distributed in $[0, 2\pi)$. We also propose using the Gamma distribution to model the amplitude component $\rho$. The Gamma PDF is defined as
\begin{equation} \label{eqn:Gamma_pdf}
    f_{\rm Gamma}(\rho; \alpha, \beta) = \frac{1}{\Gamma(\alpha)\beta^\alpha} \rho^{\alpha-1} e^{-\rho/\beta}, \quad \rho > 0,
\end{equation}
\noindent
where $\alpha > 0$ is the \textit{shape} parameter, $\beta > 0$ is the \textit{scale} parameter, and $\Gamma(\alpha)$ is the Gamma function. Adjusting \( \alpha \) and \( \beta \), we can obtain PDFs with desired shapes, such as in Fig.~\ref{fig:candidate_gamma}~(a).  An approximately flat PDF corresponds to a uniform-like distribution for sampling QAM-like constellations; a spike-like distribution corresponds to PSK-like constellations; and a distribution with balanced spread corresponds to multi-ring-like constellations. Therefore, the Gamma distribution is  capable of modeling the desired amplitude PDFs from which constellation points can be sampled. 
Combining a uniformly distributed phase \( \phi \) with a Gamma-distributed amplitude \( \rho \) yields a 2D likelihood distribution over APSK candidate points, each assigned a likelihood value as shown in Fig.~\ref{fig:candidate_gamma}~(b).

To avoid selecting candidate points that are overly clustered, we incorporate a regularization term in the sampling process that encourages inter-point separation among equally likely points. We define the regularization term as the sum of pairwise distances between the constellation points. The sampling process to select \( M \) constellation points, given parameters \(\alpha\) and \(\beta\), is formulated as
\begin{equation}
\begin{aligned}
\mathcal{S}(\alpha, \beta) &= \underset{\mathcal{S}_M \subseteq \bar{\mathcal{S}},\, |\mathcal{S}_M|=M}{\arg\max} \sum_{x_i \in \mathcal{S}_M} f(x_i; \alpha, \beta) \\
\text{s.t.} \quad &\sum_{x_i \in \mathcal{S}_M} \sum_{x_j \in \mathcal{S}_M} |x_i - x_j|^2 > \lambda
\end{aligned}
\label{eq:point_selection}
\end{equation}
\noindent
in which $\mathcal{S}_M$ is the subset with cardinality $M$, \( \lambda \geq 0 \) is a regularization coefficient that balances likelihood versus spatial diversity, and  
\( |x_i - x_j|^2 \) is the \textit{Euclidean distance} between the constellation points \( x_i \) and \( x_j \). Through \eqref{eq:point_selection}, we find the best $M$ constellation points from $\bar{\mathcal{S}}$ to form the selected set $\mathcal{S}(\alpha, \beta)$. Then, using the factor 
$\sqrt{\frac{M P}{\sum_{x_i \in \mathcal{S}_M} |x_i|^2}}$, we scale $\mathcal{S}_M$ to have average power $P$. 

Next, the communication and sensing performance metrics can be evaluated for the selected set.

\subsubsection{Communication Performance Metric}

For the communication task, we adopt the mutual information between the channel input \( x \) and output \( y \) as the performance metric. The mutual information \( I(X; Y_c) \) is defined as~\cite{empirical}
\begin{equation}
R = I(X; Y_c) = h(Y_c) - h(Y_c|X),
\end{equation}
in which \(h(\cdot)\) represents the differential entropy  and $X$ and $Y_c$ are random variables representing the input and output of the communication channel, respectively. The evaluation of \( h(Y_c|X) \) is straightforward. For example, in an AWGN channel, it has closed-form expression \( h(Y_c|X) = \log_2(\pi e \sigma_c^2) \)~\cite{proakis2008digital}. To evaluate the entropy \( h(Y_c) \), we use the marginal distribution of \( f_{Y_c}(y_c) \), to get
\begin{align}
h(Y) &= -\int f_{Y_c}(y_c) \log_2 f_{Y_c}(y_c) \, dy \notag \\
     &= -\int f_{Y_c}(y_c) \log_2 \Big( \sum_{x_i \in \mathcal{S}} f_C(y_c|x_i)p(x_i) \Big) dy_c \notag \\
     &= -\mathbb{E}_{Y_c} \Big[ \log_2 \sum_{x_i \in \mathcal{S}} f_C(y_c|x_i)p(x_i) \Big].
\label{eq:mutual_info}
\end{align}
Since \( f_{Y_c}(y_c) = \sum_{x_i \in \mathcal{S}} f_C(y_c|x_i)p(x_i) \) is a weighted sum of Gaussian distributions, there is no closed-form expression for \( h(Y_c) \). Hence, we approximate the expectation using Monte Carlo (MC) estimation
\begin{align}
h(Y_c) \approx -\frac{1}{N_{\text{MC}}} \sum_{i=1}^{N_{\text{MC}}} \log_2 \sum_{x_i \in \mathcal{S}} f_C(y_{c,i}|x_i)p(x_i),
\label{eq:montecarlo_entropy}
\end{align}

\noindent
where \( N_{\text{MC}} \) denotes the number of MC trials, and \( y_{c,i} \) represents the \( i \)-th realization of the received signal, assuming \( x_i \) was the transmitted constellation point, and \(p(x_i)=1/M\). The normalized mutual information is defined in~\eqref{eq:normalized_mi}, ensuring the metric lies in the range \([0, 1]\) for any \( M \) to be in same scale as sensing performance metrics, that is,
\begin{equation}
\bar{R} = \frac{R}{\log_2 M}.
\label{eq:normalized_mi}
\end{equation}

\begin{algorithm}[t]
\caption{: Optimization of Proposed Constellation Design}
\label{alg:isac_constellation}
\begin{algorithmic}
\REQUIRE $\bar{\mathcal{S}}$ (APSK), $M$, $\omega_d$, $[\alpha_{\min}, \alpha_{\max}]$, $ [\beta_{\min}, \beta_{\max}]$, $\lambda$, $P$, $N_{MC}$, $N_p$, and $N_{\rm iter}$
\ENSURE Optimized constellation set $\mathcal{S}^*$

\STATE Initialize best objective value: $\mathcal{F}_{\text{best}} \leftarrow -\infty$
\STATE Initialize positions and velocities of $N_p$ PSO particles with random $(\alpha, \beta)$ values within the specified bounds

\FOR{each iteration $j = 1$ to $N_{iter}$}
    \FOR{each particle $l = 1$ to $N_p$}
        \STATE Retrieve $(\alpha^{(l)}, \beta^{(l)})$ from particle $l$'s current position
        \STATE Select top $M$ points to form $\mathcal{S}$
       \STATE Scale $\mathcal{S}$ to have average power $P$
        \STATE Estimate normalized mutual information $\bar{R}$ via MC
        \STATE Compute average probability of detection $\bar{P}_d$ 
        \STATE Compute objective value: $\mathcal{F} = \omega_d \, \bar{P}_d  + (1 - \omega_d) \, \bar{R}$
        \IF{$\mathcal{F} > \mathcal{F}_{\text{best}}$}
            \STATE Update best solution: $\mathcal{S}^* \leftarrow \mathcal{S}$, $\mathcal{F}_{\text{best}} \leftarrow \mathcal{F}$
        \ENDIF
        \STATE Update PSO velocity and position for particle $l$
    \ENDFOR
\ENDFOR
\end{algorithmic}
\end{algorithm}

\subsubsection{Sensing Performance Metric}
For sensing performance, we use the probability of detection $P_d$ as the metric. Assuming a sine wave echo signal with amplitude $\rho$ in additive Gaussian noise, the envelope, $r = |y_r|$, at the receiver follows a Rice distribution, with PDF~\cite{skolnik2001radar}
\begin{equation}
f_R(r) = \frac{r}{\sigma^2_r} \exp\Big( \frac{-r^2 + \rho^2}{2\sigma^2_r} \Big) I_0\Big( \frac{r\rho}{\sigma^2_r} \Big),
\label{eq:rice_pdf}
\end{equation}

\noindent
where $I_0(\cdot)$ is the modified Bessel function of the first kind and zero order. The probability of detection ($P_d$) for a given threshold $V_T$ is given by

\begin{equation}
P_d = \int_{V_T}^{\infty} f_R(r) \, dr.
\label{eq:pd_exact}
\end{equation}

\begin{figure*}[t]
    \centering

    \includegraphics[width=0.59\linewidth]{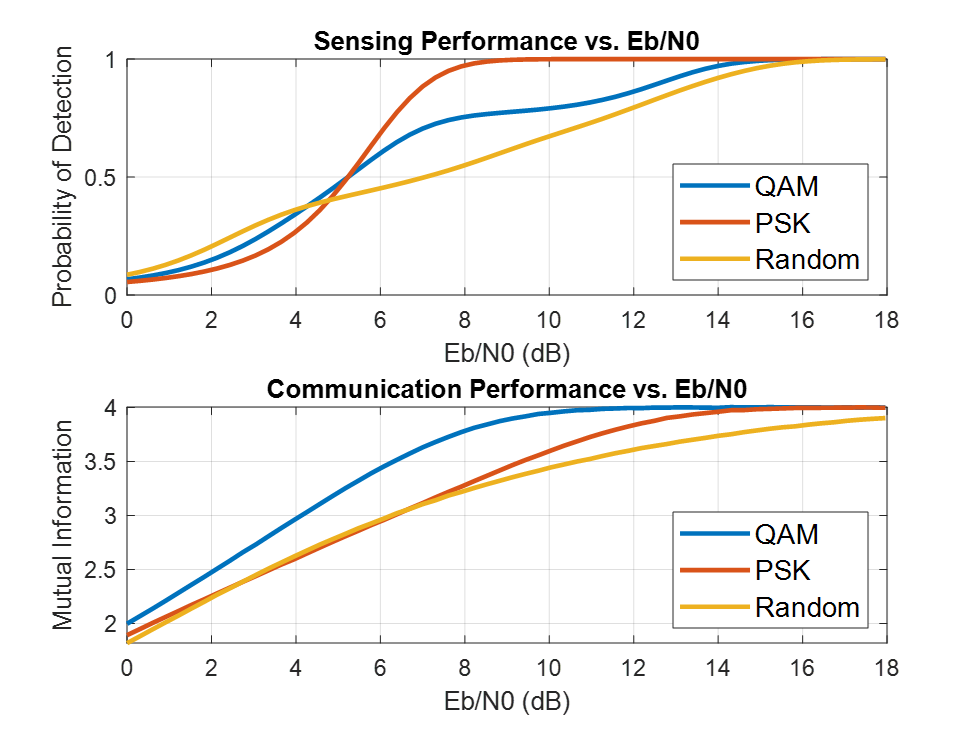}
    \hfill
    \includegraphics[width=0.4\linewidth]{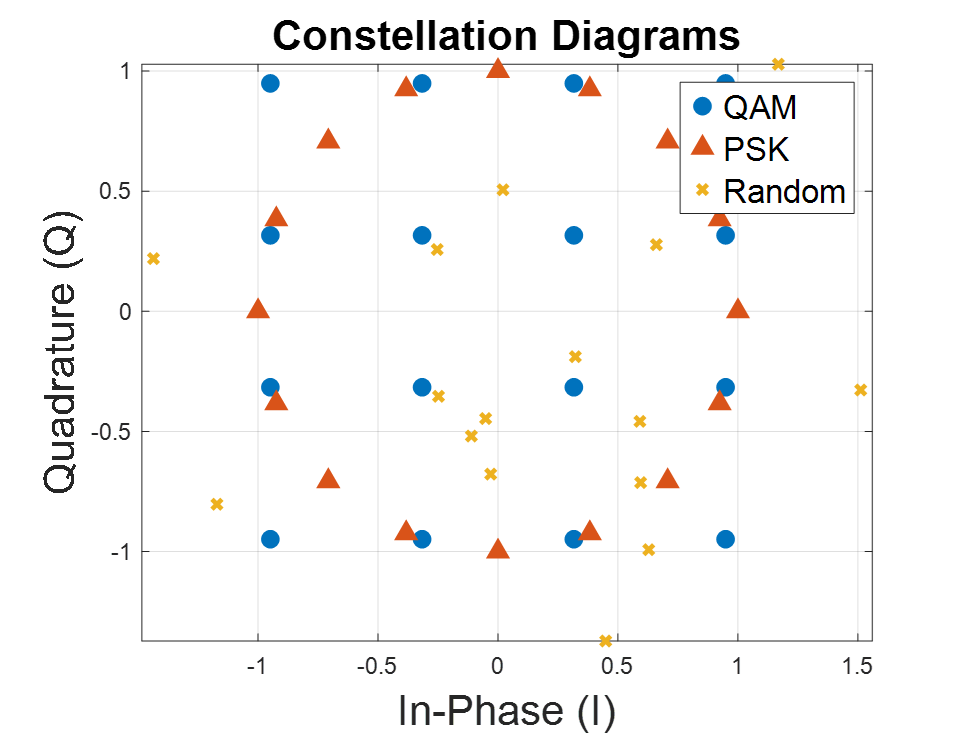}

    \par\vspace{0.2em}
    {\small(a) Probability of detection and mutual information for different constellations.\hfill
    (b) QAM, PSK, and randomly generated constellations.}

    \caption{Sensing and communication performance for three different constellations: QAM, PSK, and randomly generated.}
    \label{fig:metrics}
\end{figure*}

However, this integral cannot be evaluated in closed form. Therefore, we adopt the well-known \textit{Albersheim approximation}, which expresses $P_d$ as a function of SNR for a given false alarm probability $P_{fa}$ ,which uniqely determines the threshold $V_T$. Then, the approximate $P_d$ is given by~\cite{skolnik2001radar}
\begin{equation}
P_d = \frac{\exp( \frac{\text{SNR}_r - B}{0.12B + 1.7} )}{1 + \exp( \frac{\text{SNR}_r - B}{0.12B + 1.7} )},
\label{eq:pd_approx}
\end{equation}
where \( B = \ln\Big( \frac{0.62}{P_{fa}} \Big) \) and \(\text{SNR}_r = \frac{\rho^2}{\sigma_r^2}\). This closed-form approximation enables efficient evaluation of radar performance for a given constellation point.
We use the average probability of detection over all constellation points, defined as
\begin{align}
\bar{P}_d = \frac{1}{M} \sum_{i=1}^{M} P_{d,i},
\label{eq:avg_pd}
\end{align}
where \( P_{d,i}\) is the probability of detection corresponding to the \( i \)-th constellation point. Figure~\ref{fig:metrics} (a) shows the value of these two performance metrics for three different constellations.

\begin{rem}
APSK candidate set and Gamma distributions are chosen because they yielded the good results in ISAC. However, other choices can be made within the same framework.
\end{rem}


\subsubsection{Optimization Problem}
The final objective for the ISAC constellation optimization problem is defined as the weighted sum of the two proposed metrics: normalized mutual information and average probability of detection. The optimization is performed over the parameters of the Gamma PDF, denoted by \( \alpha \) and \( \beta \), and the objective is defined as
\begin{equation}
\begin{aligned}
\alpha^*, \beta^* = \argmax_{\alpha, \beta} \quad & \omega_d \, \bar{P}_d(\alpha, \beta)  + (1 - \omega_d) \, \bar{R}(\alpha, \beta)\\
\text{subject to} \quad & \alpha \in [\alpha_{\min}, \, \alpha_{\max}], \\
                        & \beta \in [\beta_{\min}, \, \beta_{\max}]
\end{aligned}
\label{eq:isac_objective}
\end{equation}
where \( \omega_d \in [0, 1] \) is a weight factor that controls the trade-off between communication and sensing performance. The optimization problem in~\eqref{eq:isac_objective} involves non-differentiable components, such as empirical mutual information, which prevent the use of gradient-based methods. To address this, we adopt \textit{particle swarm optimization} (PSO) which is well-suited for optimizing complex, non-smooth objective functions. In the PSO algorithm, each particle represents a candidate solution and is characterized by a position and a velocity in the search space. At each iteration, the particle updates its velocity based on the best-found positions and the dynamics of movement~\cite{kennedy1995pso}. The particle's position is then updated by adding the new velocity to its current position. The pseudo-code of the algorithm is provided in Algorithm~\ref{alg:isac_constellation}.

\section{Analytical Performance Evaluations} \label{sec:evaluation}

In this section, we evaluate the performance of the proposed Gamma-distributed constellation $\mathcal{S}$ using two metrics: (i) the \textit{union bound} \cite{proakis2008digital} on the SER for communication, and (ii) the CRB for sensing.


\subsection{Communication Performance Limit: SER Union Bound}
\label{sec:UB}

The communication channel in \eqref{eq:communication channel} is AWGN, and the receiver employs \textit{maximum-likelihood} detection as
\begin{equation}\label{eq:detection}
\hat{x} = \arg\min_{x_k \in \mathcal{S}}|y_c - x_k|^2.
\end{equation}
For two fixed constellation points $x_i$ and $x_j$, the conditional pairwise error probability (PEP) conditioned on the  transmit constellation  can be written as~\cite{proakis2008digital}
\begin{equation}\label{eq:pep}
P_{ij} = P(x_i \to x_j) = Q\!\left(\sqrt{\frac{D}{2\sigma_c^2}}\right),
\end{equation}
where $D = |d_{ij}|^2 \triangleq |x_j - x_i|^2$ and $Q(\cdot)$ is the Gaussian $Q$-function. Next, we analytically derive the average PEP over all realizations of the proposed constellation points.

 In order to derive a closed form expression for the average PEP, the PDF of $D$ in \eqref{eq:pep} must be analytically quantified. A natural first approach is to approximate the PDF of $D$ by a single Gamma distribution. Its shape ($\hat \alpha$) and scale ($\hat \beta$) parameters can be derived by using the methods-of-moments \cite{Kay1993} as  
\begin{subequations}
\begin{align}
\hat{\alpha} &= \frac{\mathbb{E}[D]^2}{\mathrm{Var}[D]} 
= \frac{2\alpha(\alpha + 1)}{(\alpha + 2)(\alpha + 3)}, \\
\hat{\beta} &= \frac{\mathrm{Var}(D)}{\mathbb{E}(D)} 
= (\alpha + 2)(\alpha + 3)\beta^2.
\end{align}
\end{subequations}
where $\hat \alpha$ and $\hat \beta$ are functions of the shape ($\alpha$) and scale ($\beta$) parameters of the Gamma distribution used to model the constellation points. 

However, while this method-of-moments approach enables a closed-form derivation of the mean and variance of $D$, it fails to capture the higher-order structure of the true distribution, particularly in the tails and near the origin, where the Q-function is most sensitive.
As a result, the SER predicted from the single-Gamma approximation can deviate significantly from the MC ground truth. 
To overcome this limitation,  we approximate the distribution of $D$ by an $L$-component mixture of Gamma distributions as  
\begin{eqnarray}\label{eqn:mixture_pdf0}
f_D(d) &\approx& \sum_{\ell=1}^{L} \omega_\ell\, f_{\mathrm{Gamma}}(d;\, \alpha_\ell, \beta_\ell) \nonumber \\ 
&=& \sum_{\ell=1}^{L} \omega_\ell \frac{d^{\alpha_\ell - 1}\,e^{-d/\beta_\ell}}{\Gamma(\alpha_\ell)\,\beta_\ell^{\alpha_\ell}}, \quad d > 0,
\end{eqnarray}
in which $L$ is the number of components. Here,  $w_\ell$, $\alpha_\ell$, and $\beta_\ell$, respectively,  are the weight, shape, and scale parameters of the $\ell$th component of the Gamma mixture,  and $\omega_\ell \geq 0$ with $\sum_{\ell=1}^{L} \omega_\ell = 1$. Then, we  present the following lemma.  
\begin{lem}
The average PEP over all realization of the transmit constellation points can be evaluated as 
\begin{equation}\label{eq:ser_gamma}
\bar{P}_{ij} \approx \frac{1}{2} \sum_{\ell=1}^{L}  w_\ell\,I_{\frac{1}{1+\gamma_\ell}}\!\left(\frac{1}{2},\,\alpha_\ell\right),
\end{equation}
in $w_\ell$, $\alpha_\ell$,  $\beta_\ell$, and $\gamma_\ell = {\beta_\ell}/{(4\sigma_c^2)}$ are the  parameters of the $\ell$th component of the Gamma mixture,  and $I_x(a,b)$ denotes the regularized incomplete beta function \cite[Eqs. (5.73)]{SimonDigital2004}. 
\end{lem}
\begin{IEEEproof}
 See Appendix \ref{app:AppendixA}.
\end{IEEEproof}

The parameters of the Gamma mixture  $\{w_\ell, \alpha_\ell, \beta_\ell\}_{\ell=1}^{L}$ are fitted by using the expectation-maximization (EM) algorithm \cite{Bishop2006}. The details of the parameter fitting through the EM algorithm is presented  in  Appendix \ref{app:AppendixB}. 

Next, for an $M$-ary constellation with  independent and identically distributed  points, every PEP is identical. Then, the \textit{union upper bound}\cite{proakis2008digital} on the average SER is obtained as follows:
\begin{equation}\label{eq:ser_ub_fixed}
\bar{P} \leq (M-1) \bar{P}_{ij}.
\end{equation}

\subsection{Sensing Performance Limit:   Channel CRB}
\label{sec:crb}


For a given transmitted signal $x_i$, the ISAC sensing receiver estimates $\zeta_r$,  as illustrated in Fig.~\ref{fig:block_diagram} using block $\psi_\zeta$. In the following, we develop an estimator for the complex reflection coefficient along with its CRB.

\begin{lem}
A \textit{maximum likelihood estimator (MLE)} of $\zeta_r$ is given by
\begin{equation}\label{eq:mle}
\hat{\zeta}_r = \frac{y_r\,x_i^*}{\rho_i^2}.
\end{equation}
Moreover, $\hat{\zeta}_r$ is an efficient \textit{minimum variance unbiased estimator} (MVUE).
\end{lem}
\begin{IEEEproof}
 See Appendix \ref{app:AppendixC}.
\end{IEEEproof}

Next, by defining a parameter vector $\boldsymbol{\eta} = [\zeta_I, \zeta_Q]^T$, where  $\zeta_I$ and $\zeta_Q$ are the in-phase and quadrature-phase components  of $\ {\zeta}_r$, the conditional Fisher information matrix for  $\boldsymbol{\eta}$ can be derived.
Then, the conditional CRB for the amplitude of the channel parameter $\zeta_r$ is (see Appendix \ref{app:AppendixC})
\begin{equation}\label{eq:crb_cond}
\mathrm{CRB}(\zeta_r|x_i) = \frac{\sigma_r^2}{\rho_i^2}.
\end{equation}
Finally, the average CRB over the probability distribution of the proposed constellation points $x_i$ is given in the following Lemma.

\begin{lem}
The average CRB for estimating the complex reflection coefficient $\zeta_r$ at the ISAC receiver is given by
\begin{equation}\label{eq:avg_crb}
\overline{\mathrm{CRB}}(\zeta_r) = \frac{\sigma_r^2}{\beta^2(\alpha-1)(\alpha-2)}, \quad \alpha > 2.
\end{equation}
\end{lem}

\begin{IEEEproof}
 See Appendix \ref{app:AppendixC}.
\end{IEEEproof}
    
The average CRB links the proposed constellation design parameters $\alpha$ and $\beta$ to sensing performance, highlighting the importance of the proposed Gamma-based design.
Under a constellation average power constraint $\mathbb{E}[\rho^2] = P$, the parameters of the Gamma distribution satisfy
\begin{equation}
\mathbb{E}[\rho^2] = \beta^2 \alpha(\alpha+1) = P,
\end{equation}
which yields $\beta^2 = \frac{P}{\alpha(\alpha+1)}.$ 
Plugging this into \eqref{eq:avg_crb}, we obtain
\begin{equation}
\overline{\mathrm{CRB}}(\zeta_r) = \frac{\sigma_r^2}{P}
\cdot
\frac{\alpha(\alpha+1)}{(\alpha-1)(\alpha-2)}, \quad \alpha > 2.
\end{equation}
This result indicates that, under a fixed average power constraint (or equivalently, for a fixed sensing channel  ${\rm SNR}_r =P/\sigma_r^2$), the average CRB depends only on the shape parameter $\alpha$, indicating that the variability of the constellation amplitudes governs sensing performance.
The function is monotonically decreasing for $\alpha > 2$, approaching $1$ as $\alpha \to \infty$ and diverging as $\alpha \to 2^+$, indicating that larger $\alpha$ leads to improved sensing performance, while values near $2$ result in degradation.

\section{Baseline Method: NN-based Constellations} \label{sec:baseline}
To provide comparison and intuition for our proposed method, we first present a data-driven approach, which uses end-to-end learning for the constellation design problem~\cite{mateos2022end}. The communication part adopts the architecture of an AE, while the radar part is trained in a supervised manner. The blue blocks (\(\psi_\epsilon, \psi_\rho, \psi_\eta\)) in Fig.~\ref{fig:block_diagram} represent an NN that maps the input to the output using trainable parameters. Table~\ref{tab:network_architecture} shows the number of neurons in the hidden  layers parametrized by  $U$ and the activation functions for each NN. 
\subsubsection{Communication NN}
For the communication receiver, we consider an autoencoder architecture and use the cross-entropy loss function as
\begin{align} \label{eq:cross entropy}
\mathcal{L}_{c} = -\mathbb{E}\Big[\sum_{j=1}^{K} \boldsymbol{m}_j \log(\hat{\boldsymbol{m}}_j)\Big],
\end{align}
where \( \boldsymbol{m}_j \) is the actual \( j \)th bit, and \( \hat{\boldsymbol{m}}_j \) is the corresponding \textit{likelihood} estimated by the NN, \( \psi_\eta \). 
\subsubsection{Sensing NN}
We use a supervised approach for sensing and define the loss functions for detection. We assume that the target is present ($T=1$) or absent ($T=0$) with a probability of \( p(T=1) = 0.5 \). The loss function for target detection is a \emph{binary cross-entropy} loss function given by
\begin{align} \label{eq:cross entropy}
\mathcal{L}_{r} = -\mathbb{E}\left[ T \log(\hat{T}) + (1 - T) \log(1 - \hat{T}) \right],
\end{align}
in which, \(\hat{T}\) represents the NN estimation by $\psi_p$. 
\subsubsection{ISAC system}
For the ISAC system, the overall loss function is the weighted sum of the communication and radar loss functions as
\begin{align}
\mathcal{L}_{\rm ISAC} = \omega_n \mathcal{L}_r + (1 - \omega_n) \mathcal{L}_c,
\end{align}
where \( \omega_n \in [0,1]\) is a hyper-parameter that defines the relative importance of communication versus radar; values of \(\omega_n\) close to 1 correspond to radar-centric scenarios, while values close to 0 correspond to communication-centric scenarios.

\begin{table}[t]
    \centering
    \caption{NN architecture details.}
    \label{tab:network_architecture}
    \renewcommand{\arraystretch}{1}
    \setlength{\tabcolsep}{8pt}
    \begin{tabular}{lccc}
        \toprule
        \textbf{Network} & \textbf{Input} & \textbf{Hidden Layers} & \textbf{Output} \\
        \midrule
        Encoder $\psi_\varepsilon$         & $K$ & $(U,\,2U,\,U)$ & $2$ (linear) \\
        Comm. Receiver $\psi_\eta$         & $2$ & $(U,\,2U,\,U)$ & $K$ (softmax) \\
        Presence Detector $\psi_p$         & $2$ & $(U,\,2U,\,U)$ & $1$ (sigmoid) \\
        \bottomrule
    \end{tabular}
\end{table}

\section{Numerical Results} \label{sec:results}
We present numerical results for the proposed framework. The message length is set to $K = 4$, yielding a constellation size of $M = 16$, and we use ${\rm SNR} = \frac{E_b}{N_0} \cdot \log_2 M$ \cite{proakis2008digital}. For the simulation of the NN baseline, the transmitted signal is normalized such that \( \mathbb{E}\left[ |x|^2 \right] = 1 \). The communication SNR is defined as \( \mathrm{SNR}_c = \ \mathbb{E}\left[ |x|^2 \right]/\sigma^2_c \), and the radar SNR is defined as \( \mathrm{SNR}_r = \mathbb{E}\left[\zeta_r^2 |x|^2 \right] / \sigma^2_r \). All simulation are done for \(\zeta_r=1\) as any other values only changes the \(\mathrm{SNR}_r\).

For training, we set the parameter of hidden layer size to \( U = 16 \) and use the rectified linear unit (ReLU) activation function between hidden layers. The model is trained using the Adam optimizer with a learning rate of 0.001 over 500 epochs, using \( 10^6 \) samples per epoch and a batch size of \( 10^5 \). All NNs are trained simultaneously. Figure~\ref{fig:constellations_all} (a) shows the constellation points designed by the NN for different values of \(\omega_n\).

For the simulation of the proposed constellation design, we set the regularization weight to \( \lambda = 0.05 \), and the number of MC samples to \( N_{MC} = 1000 \). The parameter bounds for the Gamma PDF are defined as  \([\alpha_{\min}, \alpha_{\max}] = [2, 5] \) and \( [\beta_{\min}, \beta_{\max}] = [1, 20] \). The number of PSO particles is set to \( N_p = 5 \), and the maximum number of optimization iterations is set to \( N_{\rm iter} = 1000 \). Figure~\ref{fig:constellations_all} (b) illustrates the designed constellations for different values of the trade-off factor \(\omega_d\). The probability of detection is computed using a Neyman–Pearson detector ~\cite{richards2022radar}, with \( P_{\mathrm{fa}} = 10^{-3} \).

\begin{figure*}[t]
	\centering
	\includegraphics[width=0.78\linewidth]{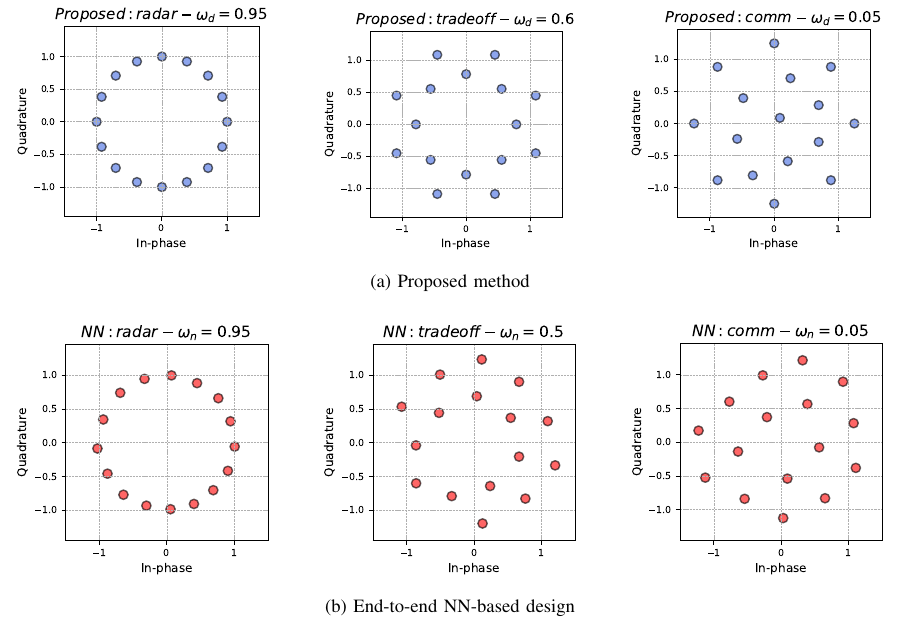}\hspace{0.03\linewidth}
	\caption{Constellation diagrams designed using the proposed method and NN–based method for different scenarios.}
	\label{fig:constellations_all}
\end{figure*}

 We set \(\alpha_{\text{min}} = 2\) to avoid placing constellation points close to the origin. The range \([\beta_{\text{min}}, \beta_{\text{max}}]\) is chosen wide enough to capture both narrow (PSK-like) and wide (QAM-like) distributions. Other hyper-parameters are selected through trial and error. Optimal $\alpha$ and $\beta$ for each $\omega_d$ are given in Table~\ref{tab:gamma_params}.  

\begin{table}[t]
    \centering
    \caption{Optimal values of the Gamma-distributed constellation parameters for different scenarios.}
    \label{tab:gamma_params}
    \renewcommand{\arraystretch}{1.15}
    \setlength{\tabcolsep}{10pt}
    \begin{tabular}{lccc}
        \toprule
        \textbf{Scenario} & $\omega_d$ & $\alpha^*$ & $\beta^*$ \\
        \midrule
        Radar-centric         & 0.95 &  4.71    &  1.00 \\
        Trade-off             & 0.60 & 3.36   &   1.00 \\
        Comm.-centric & 0.05 & 2.12 &  18.13  \\
        \bottomrule
    \end{tabular}
\end{table}

\begin{figure}[t]
    \centering
    \includegraphics[width=0.8\linewidth]{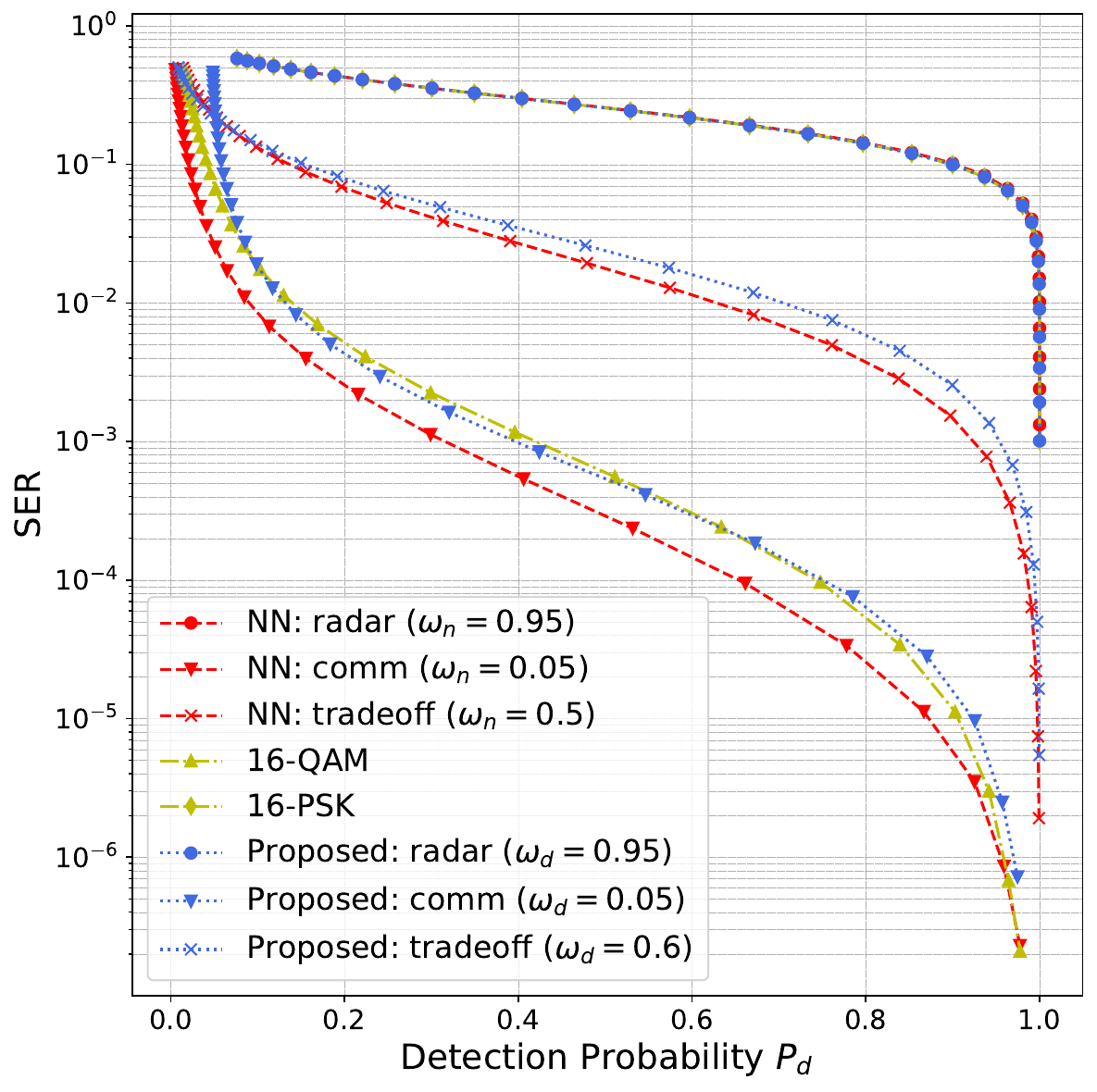}
    \caption{Trade-off between communication SER and sensing detection probability for different \(E_b/N_0\).}
    \label{fig:servspd}
\end{figure}

\begin{figure}[t]
    \centering
    \includegraphics[width=0.8\linewidth]{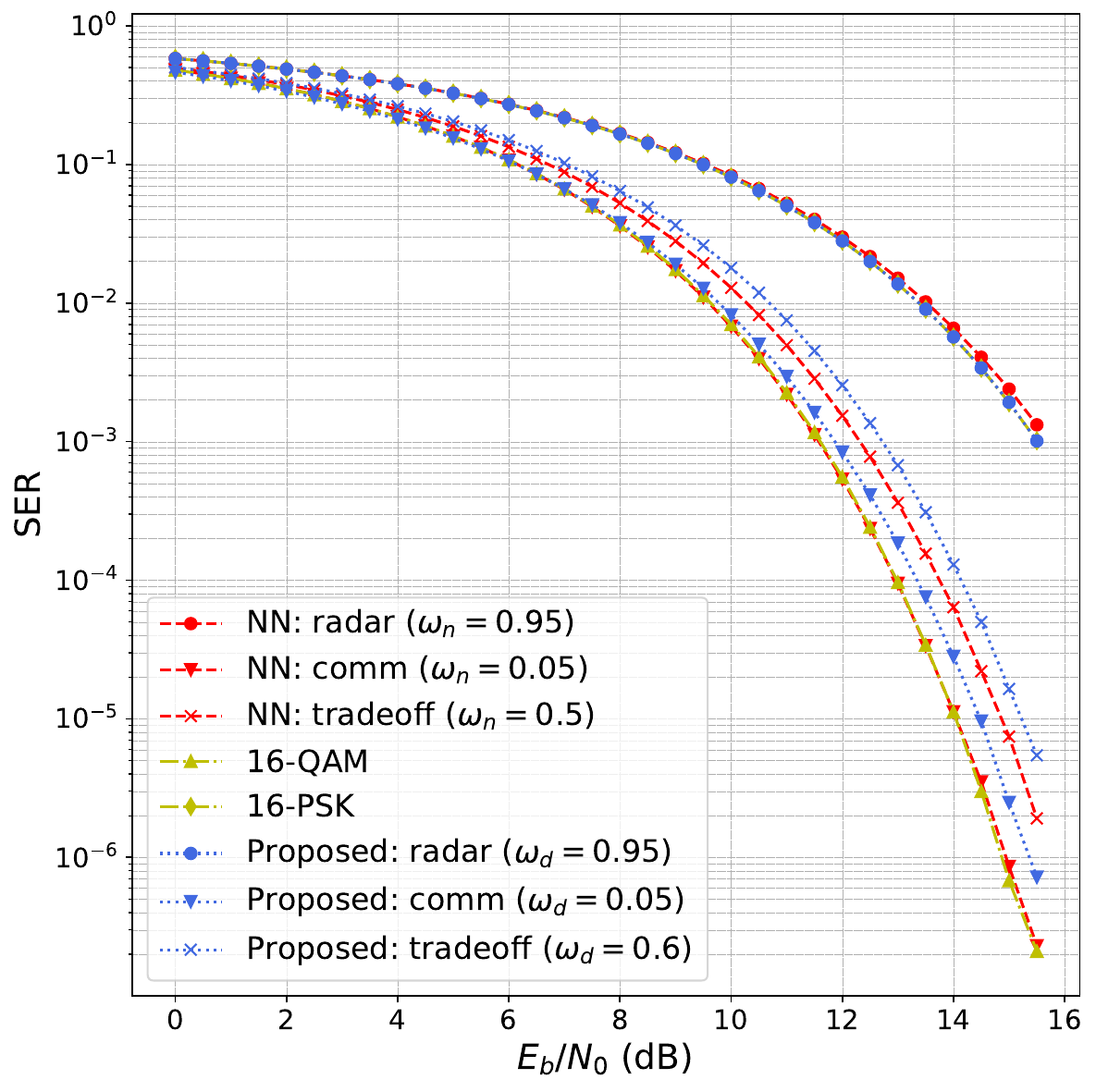}
    \caption{The symbol error rate versus  \(E_b/N_0\).}
    \label{fig:watterfall}
\end{figure}

\subsection{Comparison with NN-Based Constellations}
We compare the performance of the proposed constellation design with NN-designed constellations. 
Table~\ref{tab:scenario_comparison} presents the performance of the constellations at \( E_b/N_0 = 10~\mathrm{dB} \). The trade-off weights \( \omega_d \) and \( \omega_n \) are chosen such that \( P_d \) remains approximately the same across all scenarios, with the designed constellations illustrated in Fig.~\ref{fig:constellations_all}. As seen in each case, the proposed method achieves performance comparable to that of the NN. In the communication-centric scenario, the proposed method yields slightly worse SER, as it relies on a predefined constellation with less flexibility but implementation advantages due to its symmetric structure.

In Fig.~\ref{fig:servspd}, we plot the SER versus \( P_d \) for \( E_b/N_0 \) ranging from 0 to 15~dB with 0.5~dB steps. In this figure, each point represents a given SNR. Performance improves toward the bottom-right region, where $P_d$ is higher and SER is lower. As seen, both our method and the NN approach are capable of synthesizing PSK-like constellations, achieving relatively high \( P_d \) with low SER. In the other extreme case, our method is also able to generate constellations that maximize the minimum distance—similar to NN-based designs—achieving performance that is comparable to the QAM constellation, which offers low SER at the expense of lower \( P_d \). The trade-off scenario represents the typical operating region for ISAC systems, where both sensing and communication metrics need to be within acceptable ranges. Our method is capable of generating multi-ring constellations that strike a balance between \( P_d \) and SER, offering competitive performance to the NN-designed constellations—without requiring a labeled dataset or high computational cost, and with an analytically driven design. Figure~\ref{fig:watterfall} compares the SER curves for the different constellations.


\begin{table}[!t]
\centering
\caption{Comparison of SER and \( P_d \) in different scenarios}
\label{tab:scenario_comparison}
\renewcommand{\arraystretch}{1.2}
\begin{tabular}{llcccc}
\toprule
\textbf{Scenario} & & \textbf{Proposed} & \textbf{NN-Based} & \textbf{16-PSK} & \textbf{16-QAM} \\
\midrule
\multirow{2}{*}{Radar-centric} 
  & \textbf{SER}     & $10^{-1.09}$ & $10^{-1.07}$ & $10^{-1.09}$ & \cellcolor{gray!15} -\\
  & \(P_d\) & $0.93$       & $0.93$       & $0.93$       & \cellcolor{gray!15} -\\
\midrule
\multirow{2}{*}{Trade-off} 
  & \textbf{SER}     & $10^{-1.74}$ & $10^{-1.88}$ & \cellcolor{gray!15} - & \cellcolor{gray!15} -\\
  & \(P_d\) & $0.57$       & $0.57$       & \cellcolor{gray!15} -& \cellcolor{gray!15} -\\
\midrule
\multirow{2}{*}{Comm.-centric} 
  & \textbf{SER}     & $10^{-2.08}$ & $10^{-2.16}$ & \cellcolor{gray!15} -  & $10^{-2.15}$ \\
  & \(P_d\) & $0.14$       & $0.11$       & \cellcolor{gray!15} - & $0.16$ \\
\bottomrule
\end{tabular}

\end{table}


\begin{figure*}[t]
    \centering
    \includegraphics[width= 1.1\linewidth]{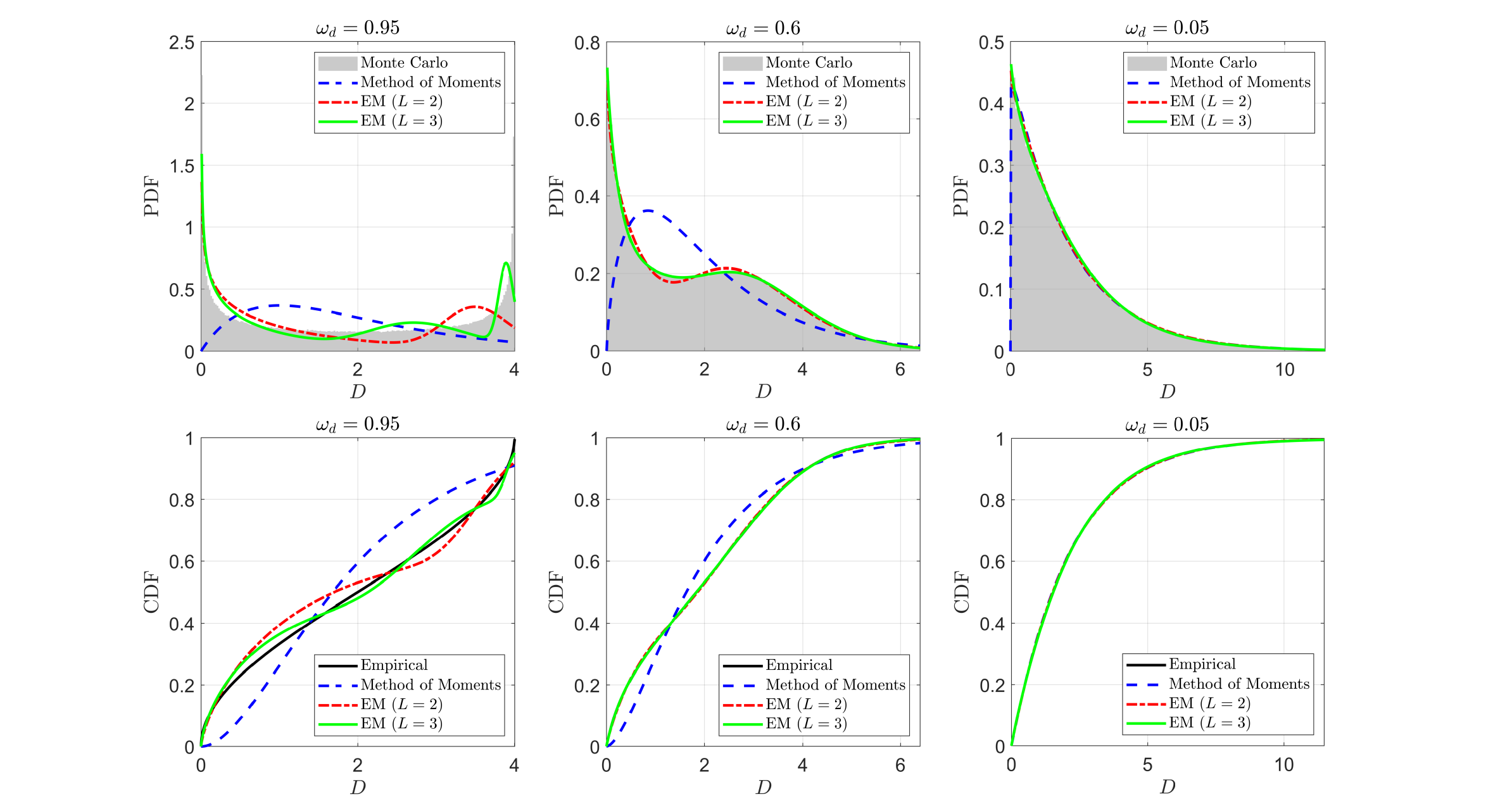}
    \caption{Verification of fitting the distribution of $D$ via a mixture of Gamma method for $\omega = \{0.05, 0.6, 0.95\}$.}
    \label{fig:pdf_cdf}
\end{figure*}

\subsection{Performance Limits of the Proposed Method}

The accuracy of the mixture-of-Gamma approximation for modeling the distribution of $D = d_{ij}^2$ (the squared distance between constellation symbols), proposed in section~\ref{sec:UB}, is validated in Fig.~\ref{fig:pdf_cdf}, where the PDF and  cumulative density function (CDF) are compared with those obtained from MC simulations. Three ISAC designs, corresponding to the proposed constellations with design weights $\omega_d = \{0.05, 0.6, 0.95\}$, are considered. In each case, the method of moments is also plotted.

For the communication-centric and trade-off designs, it is observed that both the $L = 2$ and $L = 3$ component Gamma mixtures provide an excellent fit to the empirical PDF and CDF, closely tracking the histogram across the entire support including the tail region. However, the widely-used method-of-moments clearly depicts deviations in modeling the PDF and CDF. This further emphasizes the importance of accurately fitting the distribution of $D$ using the proposed EM-based approach, as detailed in Appendix~\ref{app:AppendixB}.

For the radar-centric design ($\omega_d = 0.95$), a mismatch at the tail is observed between the fitted Gamma mixture and the empirical distribution. This degradation arises because the shape parameter is large causing the constellation amplitudes to concentrate near a constant value ($\rho_i \approx 1$, PSK-like). When the amplitude variance vanishes, the squared distance $D \approx 2(1 - \cos(\Delta\phi))$ is governed entirely by the phase difference, and its distribution converges to the Arcsine density, which is $U$-shaped. Hence, a lower number of $L$ components in the mixture fails to capture the exact tail shape. However, the accuracy of the approximation improves systematically with increasing $L$. Further increasing $L$ continues to reduce the residual mismatch by allocating additional components to approximate the boundary behavior.



\begin{figure}
    \centering
    \includegraphics[width=0.8\linewidth]{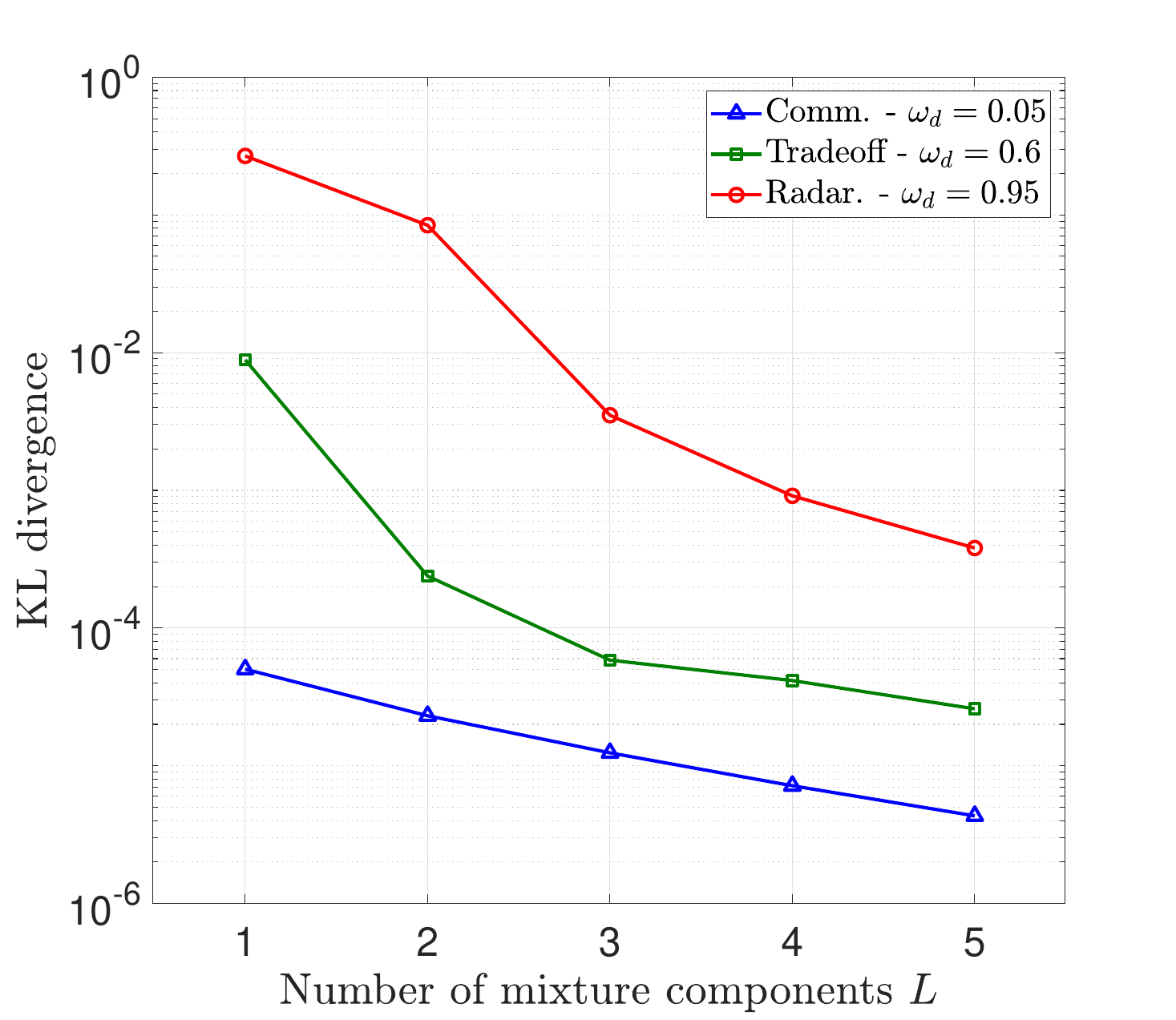}
    \caption{KL divergence for different mixture components}
    \label{fig:KL_divergence}
\end{figure}

Figure~\ref{fig:KL_divergence} evaluates the accuracy of the proposed Gamma mixture-based PDF for $D$ using the KL divergence, which measures the difference from a reference distribution. It quantifies the  amount of information lost when the PDF $\hat{f}_D$ is used to approximate
 the true PDF.
 Here, the KL divergence   between the empirical distribution of $D$ and the fitted PDF (\ref{eqn:mixture_pdf0}) is plotted as a function of the number of mixture components $L$ for the three proposed constellation designs. For all three designs, the KL divergence decreases monotonically with increasing $L$. The communication-centric and the trade-off designs achieve KL divergence below $10^{-3}$ for $L \geq 2$, indicating that the proposed mixture approximation is accurate for these cases. The radar-centric design ($\omega_d = 0.05$) exhibits a higher KL divergence ($> 10^{-3}$) for $L > 3$ reflecting the shape mismatch between the Arcsine-like true distribution and the Gamma mixture approximation, as also  discussed in Fig.~\ref{fig:pdf_cdf}. However, the KL divergence decreases below $10^{-3}$ for higher $L$ components, which render it also applicable for the radar-centric design.  For the purpose of the SER union bound presented in this work, the KL divergence analysis in Fig.~\ref{fig:KL_divergence} validates that  the proposed Gamma mixture approximation with $L = \{3,4\}$ provides accurate results.

\begin{figure}
    \centering
    \includegraphics[width=0.8\linewidth]{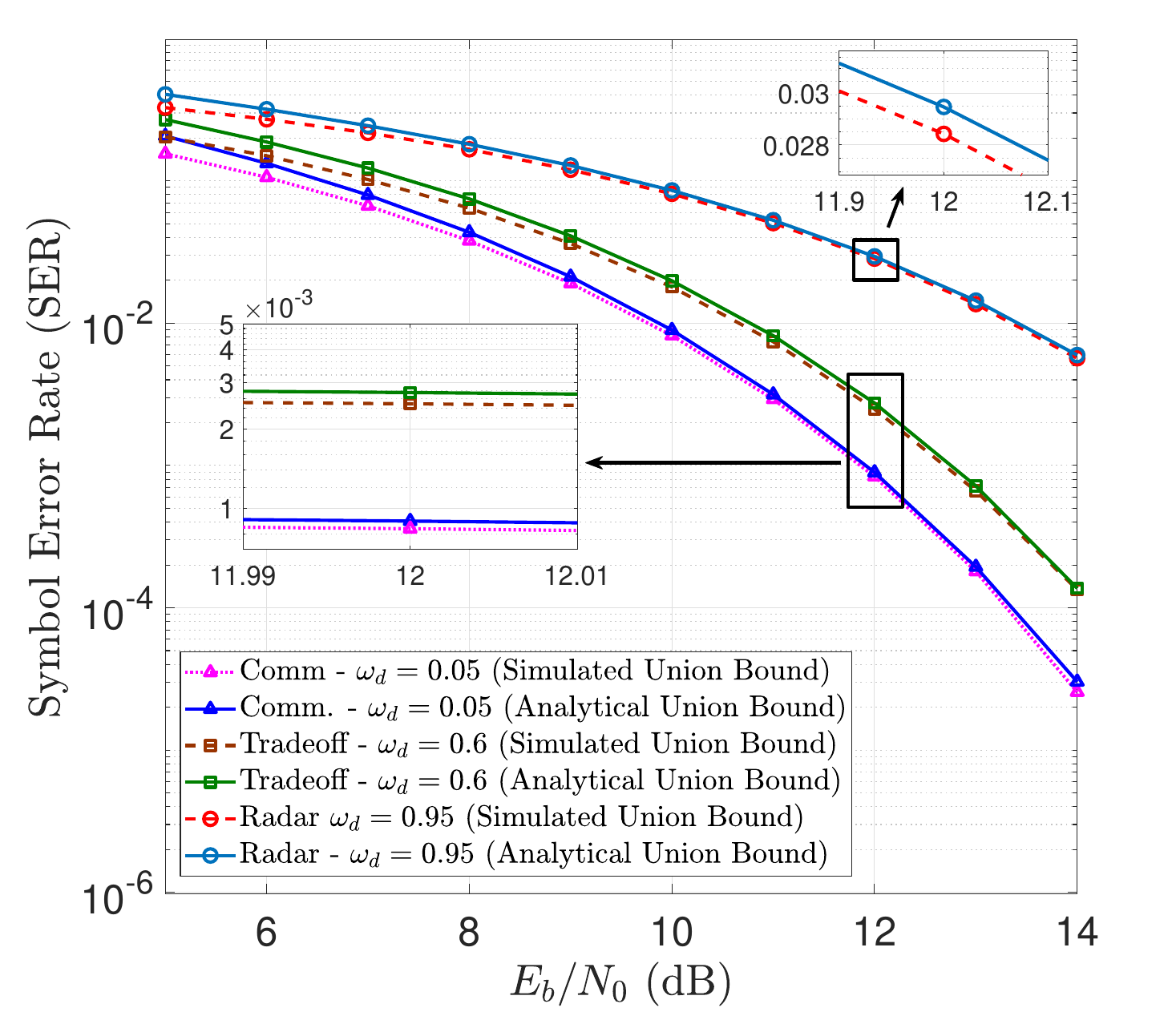}
    \caption{The average  SER   for the proposed ISAC designs.}
    \label{fig:watterfall_ub}
\end{figure}

Figure~\ref{fig:watterfall_ub} presents the average SER performance of the three designed constellations ($\omega_d = \{0.05, 0.6, 0.95\}$) as a function of $E_b/N_0$. The simulated SER curves are plotted by first averaging the  conditional PEP in (\ref{eq:pep})  over the exact distribution of the proposed constellation points generated through the MC method, and then using the   SER union bound in (\ref{eq:ser_ub_fixed}).
The simulated SER curves are compared against our closed form analytical SER union bound obtained by substituting (\ref{eq:ser_gamma}) into (\ref{eq:ser_ub_fixed}).  
For all three designs, the proposed analytical SER union bound is observed to closely track the simulated SER across the entire $E_b/N_0$ range, with the gap narrowing progressively at the higher SNR regime. 
The communication-centric design ($\omega_d = 0.05$) achieves the lowest SER owing to its larger pairwise distances, while the radar-centric design ($\omega_d = 0.95$) exhibits the highest SER due to its PSK-like geometry with smaller inter-point separation. The trade-off design ($\omega_d = 0.6$) falls between these two extremes.


In Fig.~\ref{fig:CRB}, the sensing performance of the proposed constellations is evaluated through the CRB of channel parameter estimation. The  $\overline{\mathrm{CRB}}(\zeta_r)$ is plotted against $E_b/N_0$ for the three design weights $\omega_d = \{0.95, 0.6, 0.05\}$. For each design, the simulated CRB curves are plotted by evaluating the statistical expectation operation in \eqref{eq:avg_crb_hs} over the exact distribution of the proposed constellation points through the MC method. The analytical curves are plotted by evaluating the proposed closed-form CRB expression \eqref{eq:avg_crb}. 
The variance of the MLE for $\zeta_r$ in \eqref{eq:mle} is also plotted via (\ref{eq:mse}). The simulated CRB and the variance of this MLE overlap exactly for all three  constellation designs across the entire $E_b/N_0$ regime, confirming that the MLE in (\ref{eq:mle}) is an efficient MVUE. Moreover, as a reference, the CRB of the 16-PSK constellation is also shown, which is $\overline{\mathrm{CRB}} = \sigma_r^2$ since all PSK points have identical unit amplitude, representing the best achievable CRB for any unit-power constellation. It is observed that the radar-centric design ($\omega_d = 0.95$) achieves the lowest CRB, approaching the PSK reference, as its constellation points are concentrated near the unit circle with minimal amplitude dispersion. The trade-off design ($\omega_d = 0.6$) exhibits a moderately higher CRB, while the communication-centric design ($\omega_d = 0.05$) incurs the highest CRB due to the presence of constellation points with small amplitudes that contribute large values of $\rho_i^{-2}$ to the mean.

\begin{figure}[t]
    \centering
    \includegraphics[width=1\linewidth]{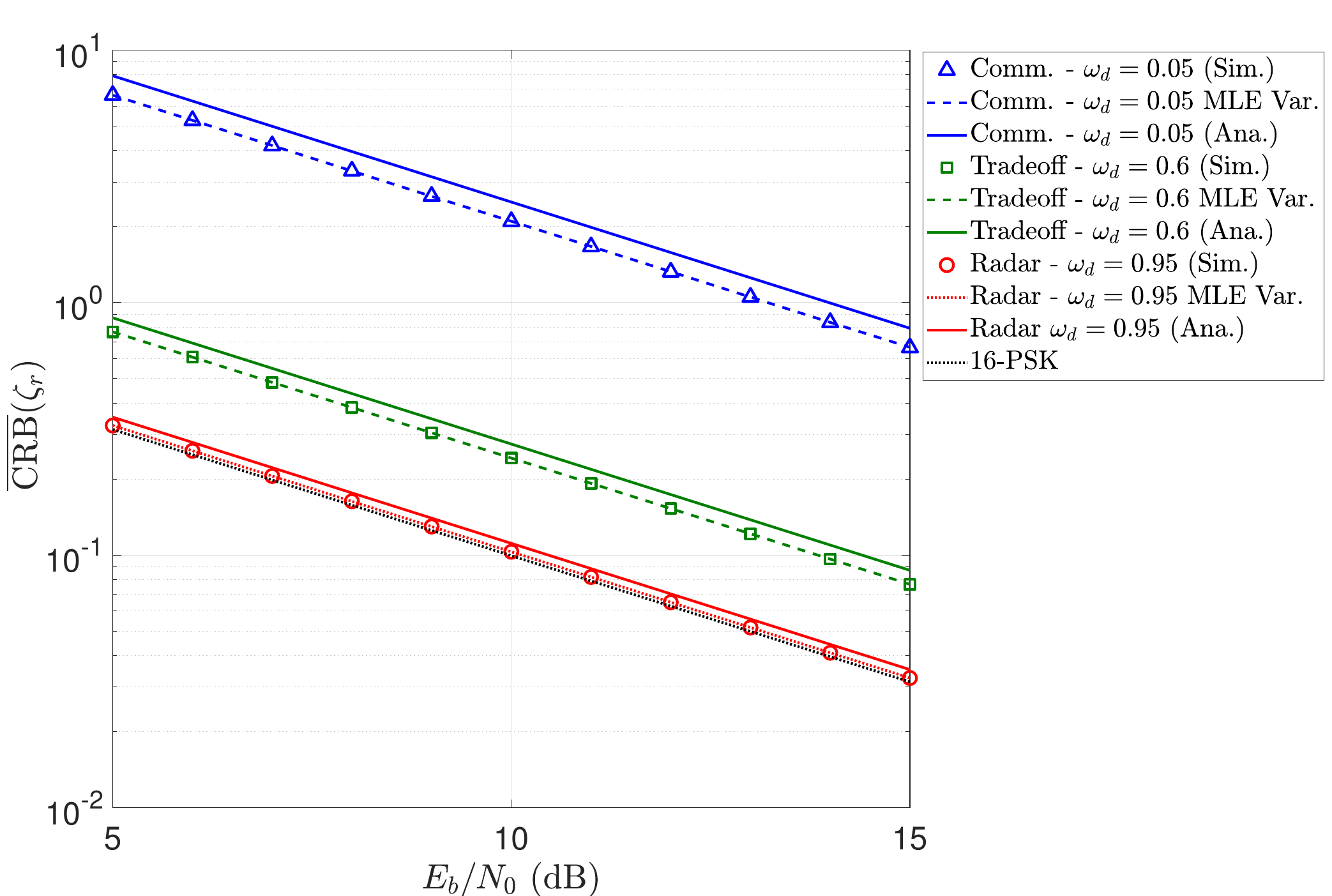}
    \caption{The CRB for estimating the channel parameter $\zeta_r$ versus $E_b/N_0$ for the proposed ISAC designs with corresponding $\alpha $ and $\beta$ listed in Table~\ref{tab:gamma_params}.}
    \label{fig:CRB}
\end{figure}

\section{Conclusion} \label{sec:conclusion}

We have proposed a novel constellation design framework in which constellation points are modeled as samples drawn from a two-dimensional PDF--Gamma-distributed in amplitude and uniformly distributed in phase. The approach formulates an optimization problem based on end-task performance metrics for both communication and sensing. The proposed method matches the performance of NN–based designs yet uses far fewer parameters and does not rely on labeled data.  By directly optimizing the 2D distribution of constellation points, our approach offers a more efficient, interpretable, and principled alternative to data-driven designs.
We derive analytical expressions for the SER upper bound in communication and the CRB for radar channel parameter estimation, and show that these bounds are reasonably tight. 
These results characterize the fundamental limits of the system and provide design insights for the proposed ISAC framework.

\appendices
\section{Derivation of SER Union Bound} \label{app:AppendixA}

We derive a closed form  expression for the upper bound of the SER when the constellation points follow the probability distributions in \eqref{eqn:Gamma_pdf}.
First, the conditional PEP is derived for a fixed constellation. Starting from \eqref{eq:communication channel},  the ML detection rule \cite{proakis2008digital} is invoked. The conditional PEP can be derived as~\cite{proakis2008digital} 
\begin{equation}\label{eqn:pep_qfunc}
P(x_i \to x_j) = Q\left(\sqrt{\frac{|d_{ij}|^2}{2\sigma_c^2}}\right),
\end{equation}
where $Q(x) =  \frac{1}{\sqrt{2\pi}} \int_x^\infty e^{-t^2/2}\,dt$ is the  Gaussian-$Q$ function. In order to derive  the average PEP,   the expected value can be taken  over the random constellation points as 
\begin{equation}
\bar{P}_{ij} = \E_{x_i, x_j} \left[Q \left(\sqrt{\frac{d_{ij}^2}{2\sigma_c^2}}\right)\right].
\end{equation}
Since the amplitudes and phases are independently distributed as  $\{\rho_i, \rho_j\} \sim \mathrm{Gamma}(\alpha,\beta)$   and $\{\phi_i, \phi_j\}\sim \mathrm{Uniform}(0, 2\pi)$, respectively,  we have
\begin{eqnarray}\label{eqn:4d_integral}
\bar{P}_{ij} = \int_0^\infty\!\!\int_0^\infty\!\!\int_0^{2\pi}\!\!\int_0^{2\pi} Q\!\left(\sqrt{\frac{d_{ij}^2(\rho_i,\rho_j,\phi_i,\phi_j)}{2\sigma_c^2}}\right) \nonumber \\ 
\times f(\rho_i)f(\rho_j)f(\phi_i) f(\phi_j)\,d\phi_i\,d\phi_j\,d\rho_i\,d\rho_j.
\end{eqnarray}
Exact evaluation of   \eqref{eqn:4d_integral} in closed form seems to be mathematically intractable.    Hence, we define a new random variable $D = d_{ij}^2$. By using the  distribution of $D$, the average PEP can be reduced to single integral as 
\begin{equation}\label{eqn:PEP_1D}
\bar{P}_{ij} = \int_0^\infty Q\!\left(\sqrt{\frac{x}{2\sigma_c^2}}\right) f_D(x)\,dx.
\end{equation}

\noindent If $f_D(d)$ can be  approximated by a probability distribution whose moment generating function  (MGF) is known in closed form, then via Craig's alternative representation of the Q-function~\cite{Craig1991}
\begin{equation}
 Q(x) = \frac{1}{\pi}\int_0^{\pi/2}\exp\!\left(-\frac{x^2}{2\sin^2\psi}\right)d\psi,
\end{equation} 
the integral in \eqref{eqn:PEP_1D} can be evaluated analytically. By first substituting the above into \eqref{eqn:PEP_1D} and then interchanging the order of integration,  $\bar{P}_{ij}$ can be rewritten as 
\begin{equation}\label{eqn:pep_mgf}
\bar{P}_{ij} = \frac{1}{\pi}\int_0^{\pi/2} M_D\!\left(-\frac{1}{4\sigma_c^2\sin^2\psi}\right)d\psi,
\end{equation}
where $M_D(s) = \E[e^{sD}]$ is the MGF of $D$.
To approximate the probability distribution of $D$, a \textit{mixture of Gamma distributions} method \cite{Atapattu2011} can be invoked. Thus,  the  probability distribution of $D$ can be approximated as an $L$-component mixture of Gamma distributions as  
\begin{eqnarray}\label{eqn:mixture_pdf}
f_D(d) &\approx& \sum_{\ell=1}^{L} \omega_\ell\, f_{\mathrm{Gamma}}(d;\, \alpha_\ell, \beta_\ell) \nonumber \\ 
&=& \sum_{\ell=1}^{L} \omega_\ell \frac{d^{\alpha_\ell - 1}\,e^{-d/\beta_\ell}}{\Gamma(\alpha_\ell)\,\beta_\ell^{\alpha_\ell}}, \quad d > 0,
\end{eqnarray}
 where $\omega_\ell \geq 0$ with $\sum_{\ell=1}^{L} \omega_\ell = 1$.  Each component has a shape parameter $\alpha_\ell > 0$ and a scale parameter $\beta_\ell > 0$. The MGF of the mixture of Gamma distributions can be defined as \cite{Atapattu2011}
\begin{equation}\label{eqn:mgf_mixture}
M_D(s) = \sum_{\ell=1}^{L} \omega_\ell\,(1 - \beta_\ell s)^{-\alpha_\ell}.
\end{equation}
By substituting \eqref{eqn:mgf_mixture} into \eqref{eqn:pep_mgf}, the average PEP can be decomposed into a weighted sum of single-component integrals as follows:
\begin{equation}\label{eq:pep_mixture}
\bar{P}_{ij} = \sum_{\ell=1}^{L} \omega_\ell\,\underbrace{\frac{1}{\pi}\int_0^{\pi/2}\!\left(\frac{\sin^2\psi}{\sin^2\psi + \gamma_\ell}\right)^{\alpha_\ell}d\psi}_{\triangleq\;\mathcal{I}(\alpha_\ell,\,\gamma_\ell)},
\end{equation}

\noindent where $\gamma_\ell \triangleq \frac{\beta_\ell}{4\sigma_c^2}$, and $\mathcal{I}(\alpha_\ell, \gamma_\ell)$ can be derived in closed form  by using  \cite[Eqs. (5.72)]{SimonDigital2004} as 
\begin{equation}\label{eq:I_beta}
\mathcal{I}(\alpha_\ell, \gamma_\ell) = \frac{1}{2}\,I_{\frac{1}{1+\gamma_\ell}}\!\left(\frac{1}{2},\, \alpha_\ell\right),
\end{equation}

\noindent where $I_x(a,b) = B(x;\,a,b)/B(a,b)$ is the regularized incomplete beta function \cite[Eqs. (5.73)]{SimonDigital2004}.  The parameters $\{\omega_\ell, \alpha_\ell, \beta_\ell\}_{\ell=1}^{L}$  can be determined via an EM-based algorithm for the mixture of Gamma distribution as explained in Appendix \ref{app:AppendixB}.

\section{EM Algorithm for Gamma Mixtures \\Parameter Fitting} \label{app:AppendixB}

The EM algorithm operates on a data-set $\mathcal{D} = \{ d^{(1)}, \cdots, d^{(N)}\}$ of $N$ independent samples drawn from the true distribution of $D$, generated by sampling \{$\rho^{(n)}_i, \rho^{(n)}_j \} \sim \mathrm{Gamma}(\alpha, \beta)$ and $\Delta \phi \sim \mathrm{Uniform}[0, 2\pi]$. Then it computes   $$d^{(n)}= (\rho^{(n)}_i)^2 + (\rho^{(n)}_j)^2 - 2\rho^{(n)}_i\rho^{(n)}_j\cos(\Delta\phi^{(n)})$$ for each $n$. The observed-data log-likelihood (LL) under the mixture of Gamma model \eqref{eqn:mixture_pdf} can be written  as
\begin{equation}\label{eq:loglik}
LL(\boldsymbol{\Theta}) = \sum_{n=1}^{N}\ln\!\left[\sum_{\ell=1}^{L} \omega_\ell\,f_{\mathrm{Gamma}}(d^{(n)};\,\alpha_\ell,\beta_\ell)\right],
\end{equation}
\noindent where $\boldsymbol{\Theta} = \{\omega_\ell, \alpha_\ell, \beta_\ell\}_{\ell=1}^{L}$. The direct maximization of \eqref{eq:loglik} is intractable due to the logarithm of a sum. Hence,  the EM algorithm first introduces a latent assignment variable for each sample indicating which component it generated, and then it iteratively alternates between computing the posterior distribution of these assignments (E-step) and updating the parameters to maximize the expected complete-data log-likelihood (M-step). Each iteration is guaranteed to increase $LL(\boldsymbol{\Theta})$, ensuring monotonic convergence to a stationary point.
In E-step, for each sample $d^{(n)}$ and each component $\ell$, the responsibility $r_{n\ell}$ measures how likely it is that the component $\ell$ generated $d^{(n)}$. Next, the responsibility metrics  can be defined via Bayes' rule as follows:
\begin{equation}\label{eq:estep}
r_{n\ell} = \frac{\omega_\ell\,f_{\mathrm{Gamma}}(d^{(n)};\,\alpha_\ell,\beta_\ell)}{\sum_{\ell'} \omega_{\ell'}\,f_{\mathrm{Gamma}}(d^{(n)};\,k_{\ell'},\theta_{\ell'})},
\end{equation}
which is evaluated in the log domain via the log-sum-exp trick to avoid numerical underflow \cite{Bishop2006}. The M-step then updates the parameters using the sufficient statistics $N_\ell = \sum_n r_{n\ell}$, $\bar{d}_\ell = N_\ell^{-1}\sum_n r_{n\ell}\,d^{(n)}$, and $\overline{\ln d}_\ell = N_\ell^{-1}\sum_n r_{n\ell}\ln d^{(n)}$, where $N_\ell$ is the effective number of samples assigned to component $\ell$, $\bar{d}_\ell$ is their weighted mean, and $\overline{\ln d}_\ell$ is their weighted log-mean, respectively.
Then, the weight and scale updates can be   derived as 
\begin{equation}\label{eq:updated_w_theta}
\omega_\ell = \frac{N_\ell}{N} \quad \text{and}\quad \beta_\ell = \frac{\bar{d}_\ell}{\alpha_\ell}. 
\end{equation}
The shape parameter update leads to a transcendental equation. We define $s_\ell = \ln\bar{d}_\ell - \overline{\ln d}_\ell$ and $\alpha_\ell$ satisfies
\begin{equation}\label{eq:shape_eq}
\ln \alpha_\ell - \psi(\alpha_\ell) = s_\ell, 
\end{equation}
where $\psi(k) = \frac{d}{dk} \ln \Gamma(k)$ is the digamma function \cite[Eq. 8.36]{gradshteyn2014table}. This equation has a unique root for each $s_\ell > 0$, and it can be  solved by Newton's method as 
\begin{equation}\label{eq:newton}
\alpha_\ell \leftarrow \alpha_\ell - \frac{\ln \alpha_\ell - \psi(\alpha_\ell) - s_\ell}{1/\alpha_\ell - \psi'(\alpha_\ell)},
\end{equation}
where $\psi'(\alpha_\ell)$ is the trigamma function. A good initialization is with the Minka approximation \cite{Minka2002} $$\alpha_\ell^{(0)} = (3 - s_\ell + \sqrt{(s_\ell-3)^2 + 24s_\ell}\,)/12s_\ell$$  and Newton's method converges with fewer iterations. 

\begin{rem}
    The scale and shape parameters are coupled as such $\beta_\ell$ depends on $\alpha_\ell$ through $\beta_\ell = \bar{d}_\ell/\alpha_\ell$. In practice, we first solve for $\alpha_\ell$ via \eqref{eq:shape_eq}-\eqref{eq:newton}, then compute $\beta_\ell$ via \eqref{eq:updated_w_theta}.
\end{rem}

The EM algorithm converges to a local maximum, and hence, precise  initialization matters. A simple and effective strategy is to sort the samples, split them into $\ell$ equal-sized groups by quantiles, and moment-match a Gamma to each group as 
\begin{equation}\label{eq:init}
\alpha_\ell^{(0)} = \frac{\hat{\mu}_\ell^2}{\hat{\sigma}^2_\ell}, \quad \beta_\ell^{(0)} = \frac{\hat{\mu}_\ell}{\hat{\sigma}^2_\ell}, \quad \text{and}\quad \omega_\ell^{(0)} = \frac{1}{L},
\end{equation}
where $\hat{\mu}_\ell$ and $\hat{\sigma}^2_\ell$ are the sample mean and variance within each group. To mitigate sensitivity to initialization, the entire EM procedure is run $R$ times with perturbed initializations, and the run that achieves the highest final log-likelihood is selected.

A fundamental property of the EM algorithm is that the observed-data LL \eqref{eq:loglik} is non-decreasing at each iteration as 
\begin{equation}
LL(\boldsymbol{\Theta}^{(\mathrm{new})}) \geq LL(\boldsymbol{\Theta}^{(\mathrm{old})}).
\end{equation}
The algorithm is terminated when the relative improvement in log-likelihood falls below a tolerance $\epsilon$ as
\begin{equation}\label{eq:stopping}
\frac{LL(\boldsymbol{\Theta}^{(\mathrm{new})}) - LL(\boldsymbol{\Theta}^{(\mathrm{old})})}{|LL(\boldsymbol{\Theta}^{(\mathrm{old})})| + \epsilon_0} < \epsilon,
\end{equation}
where $\epsilon = 10^{-8}$ and $\epsilon_0 = 10^{-16}$ are typical choices. A maximum iteration count $T_{\max}$ is also imposed as a safeguard. 

\section{Derivation of MLE and CRB for Sensing \\Channel Parameter} \label{app:AppendixC} 


In monostatic radar, the sensing receiver typically knows the transmitted symbol $x_i$, and hence, the underlying  estimation problem reduces to  estimating $\zeta_r$ given $y_r$  conditioned on $x_i$. Then, $y_r$ can be modeled as $y_r|x_i \sim \mathcal{CN}(\zeta_r x_i, \sigma_r^2)$. The MLE of $\zeta_r$ can be derived by maximizing the log-likelihood $\mathrm{ln} f (y_r|x_i)$ as 
\begin{eqnarray}\label{eq:mle_deriv}
\hat{\zeta}_r &=& \underset{\zeta_r}{\mathrm{arg\;max}} \;\; \mathrm{ln}(\pi \sigma_r^2) - \frac{|y_r - \zeta_r x_i|^2}{\sigma_r^2} \nonumber \\
&=&  \underset{\zeta_r}{\mathrm{arg\;min}} \;\; |y_r - \zeta_r x_i|^2.\label{eqn:MLE1}
\end{eqnarray}
Then, the unconstrained minimization in  (\ref{eqn:MLE1}) can be handled by taking taking the first-order derivative with respect to $\zeta_r$, and thereby, 
the MLE of $\zeta_r$ can be derived as  
\begin{equation}\label{eq:mle}
\hat{\zeta}_r = \frac{y_r\,x_i^*}{\rho_i^2}.
\end{equation}

By substituting observation \eqref{eq:radar channel} into \eqref{eq:mle} and taking expectation, we can show that this MLE  is unbiased since $\mathbb{E}[\hat{\zeta_r}] = \zeta_r$. Then, the variance of the MLE in \eqref{eq:mle}  can be derived  as
\begin{align}
\Var({\hat{\zeta}_r}) = \mathbb{E}\!\left[ \left| \hat{\zeta}_r - \zeta_r \right|^2 \right] = \E\left[\left|\frac{n_r x_i^*}{\rho_i^2}\right|^2\right] = \frac{\sigma_r^2}{\rho_i^2}. \label{eq:mse}
\end{align}
Next, we derive the CRB for estimating the sensing channel parameters. By letting $\mu \triangleq \zeta_r x_i$ and $\zeta_r = \zeta_I + j\zeta_Q$,  for a parameter vector $\boldsymbol{\eta} = [\zeta_I, \zeta_Q]^T$, the $(m,n)$th element of the Fisher information matrix  is given by  \cite{Kay1993}
\begin{align}\label{eq:fim_def}
[\mathbf{J}(\boldsymbol{\eta})]_{mn} = 2\,\mathfrak{Re} \left\{\frac{\partial\mu^*}{\partial\eta_m}\,\frac{1}{\sigma_r^2}\,\frac{\partial\mu}{\partial\eta_n}\right\}.
\end{align}
In our case, $\mu = \zeta_r x_i$ resulting in \begin{align}
\frac{\partial\mu}{\partial\zeta_I} = x_i   \quad\text{and}\quad 
\frac{\partial\mu}{\partial\zeta_Q} = jx_i. \label{eq:dmu_dtheta}
\end{align}
Next, the conditional Fisher information matrix can be derived via  \eqref{eq:fim_def}-\eqref{eq:dmu_dtheta} in closed form as 
\begin{eqnarray} 
J_{\zeta_I\zeta_I} &=& \frac{2}{\sigma_r^2}\mathfrak{Re} \left\{(x_i)^*\cdot x_i\right\} \nonumber \\ 
&=& \frac{2}{\sigma_r^2}\mathfrak{Re}\{|x_i|^2\} = \frac{2\rho_i^2}{\sigma_r^2}, \label{eq:J_xixi}\\[6pt]
J_{\zeta_Q\zeta_Q} &=& \frac{2}{\sigma_r^2}\mathfrak{Re}\!\left\{(j x_i)^*(j x_i)\right\} \nonumber \\ 
&=& \frac{2}{\sigma_r^2}\mathfrak{Re}\{|x_i|^2\} = \frac{2 \rho_i^2}{\sigma_r^2}, \label{eq:J_thth}\\[6pt]
J_{\zeta_I\zeta_Q} &=& \frac{2}{\sigma_r^2}\mathfrak{Re}\!\left\{(x_i)^*\cdot j x_i\right\} \nonumber \\ 
&=& \frac{2}{\sigma_r^2}\mathfrak{Re}\{j|x_i|^2\} = 0, \label{eq:J_xith}  \\
\mathbf{J}(\boldsymbol{\eta}| x_i) &=& \frac{2\rho_i^2}{\sigma_r^2}\begin{bmatrix} 1 & 0 \\ 0 & 1\end{bmatrix}. \label{eq:final_FIM}
\end{eqnarray}
Then, the conditional CRB  given  a fixed constellation design can be written  as 
\begin{equation}\label{eq:crb_hs}
\mathrm{CRB}(\zeta_r | x_i) = \mathrm{Tr}\left([\mathbf{J}(\boldsymbol{\eta}| x_i)]^{-1} \right) = \frac{\sigma_r^2}{\rho_i^2}= \frac{\sigma_r^2}{|x_i|^2}.
\end{equation}

By comparing (\ref{eq:mse}) and (\ref{eq:crb_hs}), we conclude that the MLE   in \eqref{eq:mle} is an efficient MVUE.

Next, we derive an average CRB over the distribution of the proposed constellation points. 
Since the transmitted symbol is random, as its amplitude and phase are drawn from the Gamma and uniform distributions, respectively,  an average  CRB can be derived by averaging it over all possible constellation realizations. The phase $\phi_i$ does not appear in $\rho_i^2 = |x_i|^2$. Thus,  this averaging is taken over the amplitude distribution only. Then, the average CRB can be derived as 
\begin{equation}\label{eq:avg_crb_hs}
\overline{\mathrm{CRB}}(\zeta_r) = \sigma_r^2\,\E[\rho_i^{-2}].
\end{equation}
\noindent The expectation term in \eqref{eq:avg_crb_hs} can be derived by using the PDF in (\ref{eqn:Gamma_pdf}) as
\begin{subequations}
\begin{eqnarray}
\E[\rho_i^{-2}] &=& \int_0^\infty \rho_i^{-2} \frac{\rho_i^{\alpha-1}e^{-\rho_i / \beta}}{\Gamma(\alpha)\beta^\alpha}\,d\rho_i  \nonumber \\
&=& \frac{1}{\Gamma(\alpha)\beta^\alpha}\int_0^\infty \rho_i^{\alpha-3}e^{-\rho_i/\beta}\,d\rho_i \label{eq:inte} \\ 
&=& \frac{\Gamma(\alpha-2)\beta^{\alpha-2} }{\Gamma(\alpha)\beta^\alpha} \label{eq:exp_rho_2}  \label{eq:inte_gamma} \\
 &=& \frac{1}{\beta^2(\alpha-1)(\alpha-2)}, \quad \alpha > 2. \label{eq:exp_rho_2}
\end{eqnarray}
\end{subequations}
where  (\ref{eq:inte_gamma}) is written by using  \cite[Eq. (8.310)]{gradshteyn2014table}  and (\ref{eq:exp_rho_2}) is obtained via \cite[Eq. (8.331)]{gradshteyn2014table}. Then, the average CRB can be   derived in closed form by substituting  \eqref{eq:exp_rho_2} into (\ref{eq:avg_crb_hs})  as shown in (\ref{eq:avg_crb}).  

\begin{rem}
   The integral in \eqref{eq:inte}, 
converges  if and only if $\alpha > 2$ \cite[Eq. (8.310)]{gradshteyn2014table}. For $0< \alpha \leq 2$, 
it diverges, making the average  CRB  approaches infinity. This establishes $\alpha > 2$   feasibility boundary of the proposed analytical solution for the CRB in (\ref{eq:avg_crb}).  
\end{rem}

\bibliographystyle{IEEEtran}

\begin{thebibliography}{10}
	\providecommand{\url}[1]{#1}
	\csname url@samestyle\endcsname
	\providecommand{\newblock}{\relax}
	\providecommand{\bibinfo}[2]{#2}
	\providecommand{\BIBentrySTDinterwordspacing}{\spaceskip=0pt\relax}
	\providecommand{\BIBentryALTinterwordstretchfactor}{4}
	\providecommand{\BIBentryALTinterwordspacing}{\spaceskip=\fontdimen2\font plus
		\BIBentryALTinterwordstretchfactor\fontdimen3\font minus
		\fontdimen4\font\relax}
	\providecommand{\BIBforeignlanguage}[2]{{%
			\expandafter\ifx\csname l@#1\endcsname\relax
			\typeout{** WARNING: IEEEtran.bst: No hyphenation pattern has been}%
			\typeout{** loaded for the language `#1'. Using the pattern for}%
			\typeout{** the default language instead.}%
			\else
			\language=\csname l@#1\endcsname
			\fi
			#2}}
	\providecommand{\BIBdecl}{\relax}
	\BIBdecl
	
	\bibitem{conference_version}
	A.~Keshavarzchafjiri and M.~Vaezi, ``Geometric constellation design for {ISAC}
	using the {Gamma} distribution,'' in \emph{Proc. IEEE Global Communications
		Conference (GLOBECOM)}, 2025, pp. 6322--6327.
	
	\bibitem{liu2022isacwaveform}
	F.~Liu, Y.~Cui, Y.~Liu, X.~Huang, C.~Yuen, Z.~Zhang, and H.~V. Poor,
	``Integrated sensing and communication waveform design: A survey,''
	\emph{IEEE Wireless Communications}, vol.~29, no.~4, pp. 144--151, 2022.
	
	\bibitem{MSE_MIMO}
	Z.~He, H.~Shen, W.~Xu, Y.~C. Eldar, and X.~You, ``{MSE-based training and
		transmission optimization for MIMO ISAC Systems},'' \emph{IEEE Transactions
		on Signal Processing}, vol.~72, pp. 3104--3121, 2024.
	
	\bibitem{vaezi2025ai}
	M.~Vaezi, G.~A. Baduge, E.~Ollila, and S.~A. Vorobyov, ``A tutorial on
	{AI}-empowered integrated sensing and communications,'' \emph{IEEE
		Communications Surveys \& Tutorials}, vol.~28, pp. 4980--5013, 2026.
	
	\bibitem{holistic}
	M.~F. Keskin, M.~M. Mojahedian, J.~O. Lacruz, C.~Marcus, O.~Eriksson,
	A.~Giorgetti, J.~Widmer, and H.~Wymeersch, ``Fundamental trade-offs in
	monostatic {ISAC: A} holistic investigation toward {6G},'' \emph{IEEE
		Transactions on Wireless Communications}, vol.~24, no.~9, pp. 7856--7873,
	2025.
	
	\bibitem{liu2018toward}
	F.~Liu, L.~Zhou, C.~Masouros, A.~Li, W.~Luo, and A.~Petropulu, ``{Toward
		dual-functional radar-communication systems: Optimal waveform design},''
	\emph{IEEE Transactions on Signal Processing}, vol.~66, no.~16, pp.
	4264--4279, 2018.
	
	\bibitem{dassanayake2025unsupervised}
	J.~K. Dassanayake, R.~Kulathunga, G.~A. Baduge, and M.~Vaezi, ``Unsupervised
	learning-based {ISAC} waveforms,'' \emph{IEEE Wireless Communications
		Letters}, vol.~14, no.~9, 2025.
	
	\bibitem{inner_bound}
	A.~R. Chiriyath, B.~Paul, G.~M. Jacyna, and D.~W. Bliss, ``Inner bounds on
	performance of radar and communications co-existence,'' \emph{IEEE
		Transactions on Signal Processing}, vol.~64, no.~2, pp. 464--474, 2016.
	
	\bibitem{du2024pcs}
	Z.~Du, F.~Liu, Y.~Xiong, T.~X. Han, Y.~C. Eldar, and S.~Jin, ``{Reshaping the
		ISAC tradeoff under OFDM signaling: A probabilistic constellation shaping
		approach},'' \emph{IEEE Transactions on Signal Processing}, 2024.
	
	\bibitem{geiger2025joint}
	B.~Geiger, F.~Liu, S.~Lu, A.~Rode, and L.~Schmalen, ``Joint optimization of
	geometric and probabilistic constellation shaping for {OFDM-ISAC} systems,''
	in \emph{Proc. IEEE 5th International Symposium on Joint Communications \&
		Sensing (JC\&S)}, 2025, pp. 1--6.
	
	\bibitem{GS_constellation_design_journal}
	Y.~Liu, Y.~Guo, Y.~Gu, M.~Wang, J.~Liu, and B.~Xia, ``Probabilistic
	shaping-based {ISAC} systems with finite constellations: {Analysis} and
	optimization,'' \emph{IEEE Transactions on Wireless Communications}, vol.~24,
	no.~12, pp. 9868--9881, 2025.
	
	\bibitem{yang2023random}
	X.~Yang, R.~Zhang, D.~Zhai, F.~Liu, R.~Du, and T.~X. Han, ``Constellation
	design for integrated sensing and communication with random waveforms,''
	\emph{IEEE Transactions on Wireless Communications}, vol.~23, no.~11, pp.
	17\,415--17\,428, Nov. 2024.
	
	\bibitem{geiger2023experimental}
	B.~Geiger, A.~Rode, and L.~Schmalen, ``An experimental validation of {ISAC}
	with probabilistic constellation shaping under {OFDM} signaling,'' in
	\emph{Proc. IEEE International Conference on Communications Workshops (ICC
		Workshops)}, Rome, Italy, May 2023, pp. 1--6.
	
	\bibitem{experimental}
	J.~Xu, Z.~Du, J.~Wang, and Y.~Xu, ``An experimental validation of {ISAC} with
	probabilistic constellation shaping under {OFDM} signaling,'' in \emph{Proc.
		IEEE International Conference on Unmanned Systems (ICUS)}, 2024, pp.
	1579--1584.
	
	\bibitem{temporal-frequency}
	Z.~Du, J.~Xu, Y.~Xiong, J.~Wang, M.~Furkan~Keskin, H.~Wymeersch, F.~Liu, and
	S.~Jin, ``Probabilistic constellation shaping for {OFDM ISAC} signals under
	temporal-frequency filtering,'' \emph{IEEE Transactions on Wireless
		Communications}, vol.~25, pp. 12\,914--12\,929, 2026.
	
	\bibitem{tang2024kl}
	S.~Tang, X.~Wang, F.~Xia, and Z.~Fei, ``Kullback-{L}eibler divergence based
	{ISAC} constellation and beamforming design in the presence of clutter,''
	\emph{ZTE Communications}, vol.~22, no.~3, pp. 4--12, Sep. 2024.
	
	\bibitem{learning_based}
	J.~Hu, K.~Han, L.~Jiang, K.~Meng, F.~Liu, and C.~Masouros, ``Learning-based
	constellation design for uplink bi-static integrated sensing and
	communication,'' \emph{IEEE Transactions on Vehicular Technology}, vol.~74,
	no.~8, pp. 13\,219--13\,224, 2025.
	
	\bibitem{mateos2022end}
	J.~M. Mateos-Ramos, J.~Song, Y.~Wu, C.~H{\"a}ger, M.~F. Keskin,
	V.~Yajnanarayana, and H.~Wymeersch, ``End-to-end learning for integrated
	sensing and communication,'' in \emph{Proc. IEEE International Conference on
		Communications (ICC)}, May 2022, pp. 1942--1947.
	
	\bibitem{proakis2008digital}
	J.~G. Proakis and M.~Salehi, \emph{Digital Communications}, 5th~ed.\hskip 1em
	plus 0.5em minus 0.4em\relax New York, NY, USA: McGraw-Hill, 2008.
	
	\bibitem{empirical}
	S.~Arimoto, ``An algorithm for computing the capacity of arbitrary discrete
	memoryless channels,'' \emph{IEEE Transactions on Information Theory},
	vol.~18, no.~1, pp. 14--20, 1972.
	
	\bibitem{skolnik2001radar}
	M.~I. Skolnik, \emph{Introduction to Radar Systems}, 3rd~ed.\hskip 1em plus
	0.5em minus 0.4em\relax New York, NY, USA: McGraw-Hill, 2001.
	
	\bibitem{kennedy1995pso}
	J.~Kennedy and R.~Eberhart, ``Particle swarm optimization,'' in \emph{Proc.
		IEEE International Conference on Neural Networks}, vol.~4, Perth, WA,
	Australia, 1995, pp. 1942--1948.
	
	\bibitem{Kay1993}
	S.~M. Kay, \emph{{F}undamentals of {S}tatistical {S}ignal {P}rocessing:
		{E}stimation {T}heory}.\hskip 1em plus 0.5em minus 0.4em\relax Upper Saddle
	River, NJ, USA: Prentice-Hall, Inc., 1993.
	
	\bibitem{SimonDigital2004}
	M.~K. Simon and M.-S. Alouini, \emph{{Digital Communication over Fading
			Channels}}.\hskip 1em plus 0.5em minus 0.4em\relax New York, NY, USA: John
	Wiley \& Sons, 2004.
	
	\bibitem{Bishop2006}
	C.~M. Bishop and N.~M. Nasrabadi, \emph{{Pattern Recognition and Machine
			Learning}}.\hskip 1em plus 0.5em minus 0.4em\relax Springer, 2006, vol.~4,
	no.~4.
	
	\bibitem{richards2022radar}
	M.~A. Richards, \emph{Fundamentals of Radar Signal Processing}, 3rd~ed.\hskip
	1em plus 0.5em minus 0.4em\relax New York, NY, USA: McGraw-Hill, 2022,
	chapter~6.
	
	\bibitem{Craig1991}
	J.~Craig, ``A new, simple and exact result for calculating the probability of
	error for two-dimensional signal constellations,'' in \emph{Proc. IEEE
		Military Communications Conference (MILCOM)}, 1991, pp. 571--575.
	
	\bibitem{Atapattu2011}
	S.~Atapattu, C.~Tellambura, and H.~Jiang, ``A mixture {Gamma} distribution to
	model the {SNR} of wireless channels,'' \emph{{IEEE} Trans. Wireless
		Commun.}, vol.~10, no.~12, pp. 4193--4203, 2011.
	
	\bibitem{gradshteyn2014table}
	I.~S. Gradshteyn and I.~M. Ryzhik, \emph{{Table of Integrals, Series, and
			Products}}.\hskip 1em plus 0.5em minus 0.4em\relax Academic press, 2014.
	
	\bibitem{Minka2002}
	T.~P. Minka, ``{Estimating a Gamma Distribution},'' \emph{Microsoft Research,
		Cambridge, UK, Tech. Rep}, 2002.
	
\end{thebibliography}


\end{document}